\documentclass{IEEEtran}%
\usepackage[ascii]{inputenc}
\usepackage[colorlinks=true,linkcolor=blue,citecolor=blue,plainpages=false,pdfpagelabels]%
{hyperref}
\usepackage{bbm,url,braket,microtype,amssymb,amsthm,mathtools,mathrsfs,cleveref,graphicx,float}
\usepackage{comment}
\usepackage[T1]{fontenc}
\usepackage[USenglish]{babel}
\usepackage[]{color}
\usepackage{tikz-cd} 

\providecommand{\U}[1]{\protect\rule{.1in}{.1in}}
\newtheorem{thm}{Theorem}\crefname{thm}{Theorem}{Theorems}
\newtheorem{lem}[thm]{Lemma}\crefname{lem}{Lemma}{Lemmas}
\crefname{prp}{Proposition}{Propositions}
\newtheorem{cor}[thm]{Corollary}\crefname{cor}{Corollary}{Corollaries}
\crefname{prb}{Problem}{Problems}
\newtheorem{dfn}[thm]{Definition}\crefname{dfn}{Definition}{Definitions}
\crefname{section}{Section}{Sections}
\crefname{appendix}{Appendix}{Appendices}
\let\oldref\ref
\renewcommand{\ref}[1]{(\oldref{#1})}
\DeclareMathOperator{\tr}{tr}
\DeclareMathOperator{\id}{id}
\DeclareMathOperator{\supp}{supp}
\newcommand{\R}{\mathbb{R}}
\newcommand{\N}{\mathbb{N}}

\newcommand{\state}[1]{{\vert #1 \rangle \langle #1\vert}}
\newcommand{\abs}[1]{{\left\vert{#1}\right\vert}}
\newcommand{\norm}[1]{{\left\Vert{#1}\right \Vert}}
\newcommand{\ot}{\otimes}
\newcommand{\red}[1]{{\textcolor{red}{[#1]}}}

\makeatletter
\def\endthebibliography{
  \def\@noitemerr{\@latex@warning{Empty `thebibliography' environment}}%
  \endlist
}
\makeatother

\begin{document}
\title{Quantum Channel Capacities Per Unit Cost}
\author{Dawei Ding,  Dmitri S. Pavlichin, and  Mark M. Wilde\thanks{Dawei Ding is with the Stanford Institute for Theoretical Physics, Stanford University, Stanford, California 94305, USA.  Dmitri S. Pavlichin is with the 
Department of Applied Physics, Stanford University, Stanford, California 94305, USA. Mark M.~Wilde is with the 
Hearne Institute for Theoretical Physics, the Department of
Physics and Astronomy, and the Center for Computation and Technology, Louisiana State
University, Baton Rouge, Louisiana 70803, USA. Copyright (c) 2017 IEEE. Personal use of this material is permitted.  However, permission to use this material for any other purposes must be obtained from the IEEE by sending a request to pubs-permissions@ieee.org.}}

\date{\today}
\maketitle

\begin{abstract}
  Communication over a noisy channel is often conducted in a setting in which different input symbols to the channel incur a certain cost. For example, for bosonic quantum  channels, the cost associated with an input state is the number of photons, which is proportional to the energy consumed. In such a setting, it is often useful to know the maximum amount of information that can be reliably transmitted per cost incurred. This is known as the capacity per unit cost. In this paper, we generalize the capacity per unit cost to various communication tasks involving a quantum channel such as classical communication, entanglement-assisted classical communication, private communication, and quantum communication. For each task, we define the corresponding capacity per unit cost and derive a formula for it analogous to that of the usual capacity. Furthermore, for the special and natural case in which there is a zero-cost state, we obtain expressions in terms of an optimized relative entropy involving the zero-cost state. For each communication task, we construct an explicit pulse-position-modulation coding scheme that achieves the capacity per unit cost. Finally, we compute capacities per unit cost for various bosonic Gaussian channels and introduce the notion of a blocklength constraint as a proposed solution to the long-standing issue of infinite capacities per unit cost. This motivates the idea of a blocklength-cost duality,  on which we elaborate in depth.
\end{abstract}

\begin{IEEEkeywords}
capacity per unit cost, bosonic Gaussian channels, quantum communication, blocklength-cost duality
\end{IEEEkeywords}

\section{Introduction}
The main concerns of information theory are determining limitations on information processing and how to attain them \cite{book1991cover}. In the task of communication over a noisy channel, for example, the usual goal is to compute the capacity of the channel, which is informally defined as the maximum number of bits that can be reliably transmitted over the channel divided by the total number of channel uses. As the number of channel uses is often directly proportional to the overall transmission time, the capacity measures the maximum rate of information transmission per unit time. Hence, the capacity is a limit to communication when given a certain \textit{time constraint}.

However, we need not restrict ourselves to time constraints. We can seek limits to communication with respect to other types of constraints, and in some practical settings these other constraints are more relevant. For example, we can imagine a setting in which a satellite deep in space is transmitting information back to Earth~\cite{pierce1981capacity}. In this case, the amount of time taken to transmit the information is not as much of a concern as is the finite amount of battery energy. Given the time and money taken to get the satellite deep into space, the receivers on Earth can afford to wait; but once the batteries are consumed, the satellite is no longer useful in the absence of an external energy source, as is often the case when deep in space. Hence, in this case, what is relevant here is an \textit{energy constraint}. This would also be relevant for deep sea communication, or in general, communication with a remote probe in a location difficult to reach. These were classic motivating examples for studying cost constraints in classical information theory, but they also naturally motivate studying cost constraints in quantum information theory. For instance, if the satellite is transmitting quantum information, a recently realized technology~\cite{yin2017satellite,liao2017satellite}, we would like to develop theory for optimizing information transmission with respect to that cost.

In general, we would like to consider constraints with respect to a certain \textit{cost} associated with transmission. In classical information theory, communication limits with respect to costs other than time were first considered in~\cite{reza1961introduction,shannon1948mathematical,pierce1978optical,golay1949note,gallager87energylimited}. Such ideas have appeared in quantum settings as well, most notably for quantum bosonic channels~\cite{GHT11,guha2011structured,guha2011quantum,DBEM11} where the relevant cost is the photon number. Now, just as the communication limit to a time constraint is the capacity, the corresponding communication limit to a general cost constraint should, informally, be the maximum amount of information that can be reliably transmitted divided by the total cost incurred in transmission. This is the \textit{capacity per unit cost}, which was introduced and extensively studied for classical channels in~\cite{verdu1990channel}. After the development of quantum information theory, capacity per unit cost was extended to channels with classical binary inputs and quantum outputs \cite{csiszar2007limit}.

Capacity per unit cost is relevant in many different settings, primarily when one is concerned with constraints other than time, as in the satellite setting mentioned above. However, it is also relevant even when one is still concerned with a time constraint. This is when \textit{input states have different time durations}. This was pointed out in~\cite{reza1961introduction,shannon1948mathematical}. Indeed, the notion of cost is very general and can appear in many different settings --- it can even be relevant for questions in quantum gravity. For instance, we could give a limit to the amount of quantum information transmitted via Hawking radiation emitted from a black hole using bounds on the mass-energy of the black hole. This would involve quantum channel capacity per unit cost, with mass-energy as the relevant cost.

In this paper, we generalize the capacity per unit cost to various communication tasks involving quantum channels, including classical communication, entanglement-assisted classical communication, private communication, and quantum communication. To do so, we first recall from classical information theory where cost was quantified with a \textit{cost function}, used in~\cite{verdu1990channel}, which associates to each input symbol a non-negative real number. In the quantum case, we employ a \textit{cost observable} $G \ge 0$, considered for cost-constrained capacities in \cite{holevo2003entanglement,H04}, in order to quantify the cost of transmitting a quantum state. Hence, any cost that can be described by a positive semidefinite observable can be considered in our framework. As mentioned above, we could consider channel uses, energy, photon number, or even a linear combination of these if that is what one is interested in, as is often the case with many practical optimization problems. The cost observable is a very natural generalization of the cost function, as it preserves two key properties. One is positive semidefiniteness, which we enforce by requiring $G \ge 0$. Another is additivity across channel uses. This was implicit in the definition of the classical cost function, that the total cost incurred across multiple channel uses is the sum of the costs incurred in individual channel uses. We can enforce this by defining the cost observable for an input state to $n$ channel uses to be
\begin{equation}
  G_n \equiv \sum_{j=1}^{n} I^{\ot j-1} \ot G \ot I^{\ot n-j}. \label{eq:cost-observable}
\end{equation}
By inspection, the total cost will be additive across the channel uses. Note that the classical cost function can be embedded into a cost observable by letting the spectrum of the latter be the image of the former.

Using this prescription, we can then generalize the results in~\cite{verdu1990channel} to classical communication over quantum channels by giving a formula for the capacity per unit cost of the form
\begin{equation}
  \sup_{\{p_X, \rho_A(x)\}} \frac{I(X;B)_{\rho}}{\tr[G \bar \rho_A]},\label{eq:cap-per-unit-cost-1st-time}
\end{equation}
where $A,B$ denote quantum systems, $\bar \rho_A \equiv \sum_{x}^{} p_X(x) \rho_A(x)$, and the mutual information
$I(X;B)$ is evaluated with respect to a classical--quantum state $\rho_{XB}$ associated with the channel. For simplicity, we focus on  channels with  a single-letter cost-constrained classical capacity, including entanglement-breaking channels for example \cite{S02,Shirokov2006,holevo2008entanglement}, but the developments easily generalize beyond this case and we discuss this point later on. The formula in \cref{eq:cap-per-unit-cost-1st-time} is derived via the known formula \cite{H04} for the cost-constrained classical capacity, also known as the classical capacity cost function, of a quantum channel. This is the capacity with an average cost constraint over all the channel uses. Hence, the ratio of the cost-constrained capacity to the average cost constraint is achievable as a capacity per unit cost. Conversely, this is the highest possible rate per unit cost since any higher rate would imply that we could achieve a higher rate per channel use than the cost-constrained capacity. We can then write the cost-constrained capacity as an optimized mutual information, and thus~\cref{eq:cap-per-unit-cost-1st-time} follows. Note that the formula reduces to the regular capacity when $G$ is the identity operator. This is intuitive since in this case the cost is the number of channel uses, and so every quantum state incurs unit cost.

Now, more interesting results come about in the special case in which there is a zero-cost quantum state, that is, some state $\psi^0$ such that $\tr[G \psi^0] = 0$. This is a natural setting to consider, given that transmitting a zero-cost state often physically corresponds to not actively sending anything through the channel. For example, for a bosonic channel, the zero-cost state is the vacuum. Now, by the positive semi-definiteness of $G$, without loss of generality we can take $\psi^0$ to be pure. In this case, we find that the capacity per unit cost reduces to the following expression:
\begin{equation}
  \sup_{\psi \neq \psi^0} \frac{D(\mathcal{N}(\psi) \Vert \mathcal{N}(\psi^0))}{\bra{\psi} G \ket{\psi}},
\end{equation}
where $D$ denotes the quantum relative entropy \cite{U62} and the supremum is with respect to pure states $\psi$. The above expression is intuitive given the fact that pulse-position-modulation (PPM) protocols achieve the capacity per unit cost when there is a zero-cost state~\cite{verdu1990channel}. Such PPM protocols encode information into the position of a $\psi$-pulse amidst a baseline of zero-cost $\psi^0$ states. Hence, for these protocols, we expect the relevant variables for computing the capacity per unit cost to be the distinguishability of the states $\mathcal{N}(\psi)$ and $ \mathcal{N}(\psi^0)$ in addition to the cost of $\psi$.
We then extend these results to various other communication tasks over a quantum channel, such as entanglement-assisted, private, or quantum communication. See Sections~\ref{sec:EA-comm}, \ref{sec:priv-comm}, and \ref{sec:quant-comm} for details. We apply these formulas to various quantum Gaussian channels in \cref{sec:gaussian}. In \cref{sec:blockConstraint} we introduce the notion of a blocklength-constrained capacity per unit cost, analogous to that of a cost-constrained capacity, and we derive a formula for it. We find that a blocklength constraint can ensure that the capacity per unit cost is finite and thus can play a similar role to a cost constraint for the usual capacity. This motivates the notion of a \textit{blocklength-cost duality}, which we develop with various examples and concepts.

\textit{Related Work}: After deriving many of the results in this paper and while drafting this manuscript, a related work appeared on the quant-ph arXiv \cite{J17}. In \cite{J17}, the author considers classical capacity per unit cost for particular channels that accept a general classical input symbol and  output a quantum state (known as classical--quantum channels in the literature), thus generalizing the approach in~\cite{csiszar2007limit}. In particular, a cost function is considered to quantify the cost of classical input symbols. We note here that our paper generalizes this setup to the fully quantum case in which there is a cost observable and the channels considered have quantum inputs and quantum outputs. 

\section{Preliminaries}

For simplicity, we restrict our developments to finite-dimensional Hilbert spaces, with the exception of Section~\ref{sec:gaussian}, which applies to quantum Gaussian channels. Let $\mathcal{H}_A$ and $\mathcal{H}_B$ denote finite-dimensional Hilbert spaces, and let $\mathcal{L}(\mathcal{H}_A)$ and $ \mathcal{L}(\mathcal{H}_B)$ denote spaces of linear operators acting on those respective Hilbert spaces. We denote by $\mathcal{N}_{A \to B}: \mathcal{L}(\mathcal{H}_A) \to \mathcal{L}(\mathcal{H}_B)$ a quantum channel, defined to be a completely positive and trace-preserving map. By Stinespring's dilation theorem~\cite{stinespring1955positive,paulsen2002completely}, $\mathcal{N}_{A \to B}$ can be extended to an isometric channel $\mathcal{U}_{A \to BE}: \mathcal{L}(\mathcal{H}_A) \to \mathcal{L}(\mathcal{H}_B \ot \mathcal{H}_E)$, where $\mathcal{H}_E$ is some other finite-dimensional Hilbert space and
$\mathcal{N}_{A \to B} = \tr_E \circ \ 
\mathcal{U}_{A \to BE}$.

Now, let $G \in \mathcal{L}(\mathcal{H}_A)$ be a positive semi-definite operator acting on $\mathcal{H}_A$. Throughout, we refer to $G$ as the {\it cost observable}. This is the standard cost constraint used in applications of quantum Shannon theory (see, for instance,~\cite{holevo2003entanglement,H04,holevo2008entanglement,holevo2013quantum}). As mentioned above, this is also a quantum generalization of a classical cost function~\cite{verdu1990channel}, which is a map from the input alphabet to the non-negative reals. However, note that unlike in the classical case, we can use quantum codewords that are not eigenstates of the cost observable. This might even be necessary to achieve the capacity of a quantum channel. For example, for single-mode phase-insensitive bosonic Gaussian channels, the relevant cost observable is the photon number operator, but it is known that coherent states, not number states, achieve the classical capacity~\cite{giovannetti2004classical,GHG15,giovannetti2014ultimate}.

Lastly, given two quantum states $\rho, \sigma \in \mathcal{S}(\mathcal{H}_A)$, where $\mathcal{S}(\mathcal{H}_A) \subsetneq \mathcal{L}(\mathcal{H}_A)$ denotes the set of positive semi-definite operators with unit trace, a quantum hypothesis test with $N$ copies is a binary positive-operator valued measure (POVM) $\left\{ \Lambda_N, I- \Lambda_N \right\}$ that distinguishes between $N$ copies of the two states. The two states to be distinguished are called the null and alternative hypotheses, respectively. Now, there are two possible errors that can occur. Taking $\Lambda_N$ to be the measurement result that declares the state to be $\rho^{\otimes N}$, the error probabilities are given by
\begin{align}
  & \alpha_N(\Lambda_N) \equiv \tr[(I - \Lambda_N) \rho^{\ot N} ],\\
  & \beta_N(\Lambda_N) \equiv \tr[\Lambda_N\sigma^{\ot N} ].
\end{align}
The errors are called the Type~I and Type~II errors, respectively. Then, for some $\varepsilon \in (0,1)$, we can define the following quantity:
\begin{equation}
  \label{eq:minType2}
  \beta^*_N(\varepsilon) \equiv \inf_{\Lambda_N} \left\{ \beta_N(\Lambda_N) \vert \alpha_N(\Lambda_N) \le \varepsilon \right\}.
\end{equation}
That is, it is the lowest Type~II error possible, given that the Type~I error does not exceed $\varepsilon$. By Quantum Stein's Lemma~\cite{hiai1991proper,ogawa2000strong}, for all $\varepsilon\in(0,1)$, 
\begin{align}
  \label{eq:quantumStein}
  \lim_{N \to \infty} -\frac{1}{N} \log_2 \beta_{N}^*(\varepsilon) & =  D(\rho \Vert \sigma),
\end{align}
where the quantum relative entropy 
$D(\rho \Vert \sigma)$ is defined as \cite{U62}
\begin{equation}
  D(\rho \Vert \sigma)  \equiv \tr[\rho(\log_2 \rho - \log_2 \sigma)]
\end{equation}
whenever $\operatorname{supp}(\rho)\subseteq 
\operatorname{supp}(\sigma)$ and it is equal to $+\infty$ otherwise. 

\section{Classical Communication}

We first consider the case of unassisted classical communication over a quantum channel $\mathcal{N}$. Let $n, M \in \N$, $\nu \in \R_{> 0}$, and $\varepsilon \in [0,1]$. We denote an $(n, M, \nu, \varepsilon)$ code as one with blocklength $n$ and number of messages $M$. Furthermore, denoting the quantum codewords as the density operators $\rho_{A^n}(1),\dots, \rho_{A^n}(M) \in \mathcal{S}(\mathcal{H}_A^{\ot n})$, each quantum codeword satisfies
\begin{equation}
  \tr[ G_n \rho_{A^n}(x)] \le \nu.
\end{equation}
Finally, given that the decoder uses a POVM $\{\Pi_1, \dots, \Pi_M\}$ to guess the message, the average error probability over the possible messages cannot exceed $\varepsilon$, i.e., 
\begin{equation}
  \frac{1}{M} \sum_{m=1}^{M} \tr[(I^{\otimes n} - \Pi_m)\mathcal{N}^{\ot n}(\rho_{A^n}(m)) ]
\leq \varepsilon.
\end{equation}

We recall the definition for the cost-constrained classical capacity of a quantum channel \cite{holevo2003entanglement,H04}:
\begin{dfn}
  \label{dfn:constrainedCapacity}
  Given $\varepsilon \in [0,1)$, and $\beta >0$, a non-negative number $R$ is an $\varepsilon$-achievable rate with average cost not exceeding $\beta$ if for all $\delta >0$, $\exists \, n_0 \in \N$ such that if $n \ge n_0$, then $\exists$ an $(n, M, n \beta, \varepsilon)$ code such that $\frac{\log_2 M}{n} > R-\delta$. Then, $R$ is called achievable if it is $\varepsilon$-achievable for all $\varepsilon \in (0,1)$. The supremum of all achievable rates with average cost not exceeding $\beta$ is denoted $C(\mathcal{N},\beta)$, the classical capacity cost function.
\end{dfn}

For simplicity, for the rest of the section we will consider channels with additive Holevo information at all cost constraints, i.e., 
\begin{equation}
  \forall \beta \ge 0, n \in \mathbb{N}, \quad \chi(\mathcal{N}^{\ot n}, n\beta) = n \chi(\mathcal{N}, \beta),
\end{equation}
where
\begin{equation}
  \chi(\mathcal{N}^{\ot n}, n \beta) \equiv \sup_{\substack{ \left\{ p_X, \rho_{A^n}(x) \right\}\\ \tr[G_n \bar \rho_{A^n}] \le n\beta}} I(X;B^n)_{\rho}
\end{equation}
and 
\begin{align}
  \rho_{XB^n} &= \sum_x p_X(x) \vert x \rangle \langle x \vert_X \otimes \mathcal{N}^{\otimes n}(\rho_{A^n}(x)),\\
  \bar \rho_{A^n} & = \sum_{x}^{} p_X(x) \rho_{A^n}(x).
\end{align}
Similar to the classical case \cite{verdu1990channel}, the classical capacity cost function can be computed as an optimization of a mutual information with respect to input ensembles that satisfy the cost constraint:
\begin{thm}
[\cite{schumacher1997sending,holevo1998capacity,holevo2003entanglement,H04}]
  For a channel with additive Holevo information at all cost constraints, the classical capacity cost function is given by
  \begin{align}
    \label{eq:capacityCost}
    C(\mathcal{N},\beta)  = \chi(\mathcal{N}, \beta) \equiv \sup_{\substack{ \left\{ p_X, \rho_{A}(x) \right\}\\ \tr[G \bar \rho_{A}] \le \beta}} I(X;B)_{\rho}.
  \end{align}
\end{thm}
We now give the definition of the classical capacity per unit cost. 
\begin{dfn}
  \label{dfn:capacityPerCost}
  Given $\varepsilon \in [0,1)$, a non-negative number $\textbf{\textit{R}}$ is an $\varepsilon$-achievable rate per unit cost if for every $\delta > 0$, $\exists \, \nu_0 > 0$ such that if  $\nu \ge \nu_0$, $\exists$ an $(n, M, \nu, \varepsilon)$ code such that $\log_2 M > \nu (\textbf{\textit{R}}- \delta)$. $\textbf{\textit{R}}$ is achievable if it is achievable for all $\varepsilon \in (0,1)$ and the capacity per unit cost is the supremum of all achievable rates per unit cost, denoted as $\textbf{\textit{C}}(\mathcal{N})$.
\end{dfn}
Observe that the above definition is similar to that of the usual capacity, except that we replace the blocklength $n$ by the cost $\nu$. 
In fact, we can also give an expression for the classical capacity per unit cost in terms of an optimized mutual information:
\begin{thm}
  \label{thm:generalCCost}
  The capacity per unit cost for a channel with  additive Holevo information at all cost constraints is given by
  \begin{equation}
    \textbf{\textit{C}}(\mathcal{N}) = \sup_{\beta > 0} \frac{C(\mathcal{N}, \beta)}{\beta} = \sup_{ \left\{ p_X, \rho_A(x) \right\}} \frac{I(X;B)_\rho}{\tr[G \bar \rho_A]}.
  \end{equation}
\end{thm}

The proof of Theorem~\ref{thm:generalCCost} is based on~\cite{verdu1990channel} and follows from the achievability and converse for the cost-constrained classical capacity. In general it is sufficient to prove the coding theorem for the cost-constrained capacity in order to establish a coding theorem for the capacity per unit cost.
\begin{proof}[Proof of Theorem~\ref{thm:generalCCost}]
  We first show the achievability statement
  $\textbf{\textit{C}}(\mathcal{N}) \geq \sup_{\beta > 0} \frac{C(\mathcal{N}, \beta)}{\beta}$.
   Let $\beta >0$. Let $R$ be an achievable rate per channel use with average cost not exceeding $\beta$. We claim that $R/\beta$ is an achievable rate per unit cost. This is clear for $R = 0$, and so we can assume $R>0$. To see the claim, let $\varepsilon \in (0,1)$ and fix some $\delta > 0$. Then, by definition $\exists\, n_0$ such that for all $n \ge n_0$ there is an $(n, M, n \beta, \varepsilon)$ code such that
  \begin{equation}
    \frac{\log_2 M}{n} > R - \frac{\beta \delta}{2}.
  \end{equation}
  This same code is an $(n, M, n\beta, \varepsilon)$ code such that
  \begin{equation}
    \frac{\log_2 M}{n \beta} > \frac{R}{\beta} - \frac{\delta}{2}.
  \end{equation}
  Now, let $\nu_0 = \max\left\{ (n_0+1) \beta, \frac{2 R}{\delta} \right\}$ and $\nu \ge \nu_0$. If $\nu = n\beta$ for some $n \in \N$, then $n \ge n_0$, and so the above $(n, M, n\beta, \varepsilon)$ code satisfies the necessary requirements. If instead $n \beta < \nu < (n+1)\beta$ for some $n \in \N$, then we note
  \begin{equation}
    n+1 > \frac{\nu}{\beta} \ge \frac{\nu_0}{\beta} \ge \frac{2 R}{\delta \beta},
  \end{equation}
  and so 
  \begin{align}
    \left( \frac{R}{\beta} - \frac{\delta}{2} \right) \frac{n \beta}{\nu} & > \left( \frac{R}{\beta}-\frac{\delta}{2} \right) \frac{n}{n+1}\\
    & = \left( \frac{R}{\beta}-\frac{\delta}{2} \right) \left( 1-\frac{1}{n+1} \right)\\
    & > \left( \frac{R}{\beta}-\frac{\delta}{2} \right) \left( 1- \frac{\delta\beta}{2R } \right)\\
    & > \frac{R}{\beta} - \delta.
  \end{align}
  Now, $\left( n+1 \right)\beta> \nu \ge \nu_0$, so $n \ge n_0$. Hence, the above $(n, M, n\beta, \varepsilon)$ code is a $(n, M, \nu, \varepsilon)$ code such that
  \begin{align}
   \frac{\log_2 M}{\nu} > \left( \frac{R}{\beta} - \frac{\delta}{2}  \right)\frac{n \beta}{\nu} > \frac{R}{\beta} -\delta.
  \end{align}
  Hence we have shown achievability.

  We next show the converse statement 
  $\textbf{\textit{C}}(\mathcal{N}) \leq \sup_{\beta > 0} \frac{C(\mathcal{N}, \beta)}{\beta}$.
  Suppose $\mathcal{N}$ has an $(n, M, \nu,\varepsilon)$ code. By a standard data-processing argument and entropy continuity bound~\cite{wilde2013quantum,holevo2013quantum,winter2016tight}, we have
  \begin{equation}
    \log_2 M \le \chi(\mathcal{N},\nu) + f(M, \varepsilon).
  \end{equation}
  where
  $f(M, \varepsilon) = \varepsilon \log_2 M
  + (\varepsilon+1)\log_2(\varepsilon+1)-
  \varepsilon \log_2 \varepsilon$, so that
  $\lim_{\varepsilon \to 0} f(M, \varepsilon) = 0$. Thus,
  \begin{align}
    \frac{\log_2 M}{\nu} & \le \frac{\chi(\mathcal{N},\nu)}{\nu} + \frac{f(M,\varepsilon)}{\nu}\\
    & = \frac{1}{\nu}\sup_{\substack{ \left\{ p_X, \rho_{A}(x) \right\}\\ \tr[G \bar \rho_{A}] \le \nu}} I(X;B)_{\rho} + \frac{f(M,\varepsilon)}{\nu}\\
    & \le \sup_{\beta >0} \frac{1}{\beta}\sup_{\substack{ \left\{ p_X, \rho_{A}(x) \right\}\\ \tr[G \bar \rho_{A}] \le \beta}} I(X;B)_{\rho} + \frac{f(M,\varepsilon)}{\nu}\\
    & =  \sup_{\beta >0} \frac{C(\mathcal{N},\beta)}{\beta} + \frac{f(M,\varepsilon)}{\nu}.
  \end{align}
  Thus, for any $\varepsilon$-achievable rate per unit cost $\textbf{\textit{R}}$, for any $\delta > 0$ there exists $\nu_0$ such that for $\nu \ge \nu_0$,
  \begin{equation}
    \textbf{\textit{R}} - \delta < \sup_{\beta > 0} \frac{C(\mathcal{N},\beta)}{\beta} + \frac{f(M,\varepsilon)}{\nu}.
  \end{equation}
  Hence, for all $\delta >0$,
  \begin{equation}
    \textbf{\textit{R}} - \delta < \liminf_{\nu \to \infty} \left( \sup_{\beta > 0} \frac{C(\mathcal{N},\beta)}{\beta} + \frac{f(M,\varepsilon)}{\nu}\right).
  \end{equation}
  Hence, if \textbf{\textit{R}} is an achievable rate per unit cost, then
  \begin{equation}
    \textbf{\textit{R}} \le \sup_{\beta > 0} \frac{C(\mathcal{N},\beta)}{\beta}.
  \end{equation}
  This establishes the first equality in \cref{thm:generalCCost}.

  To show the second equality, we first argue
  \begin{align}
    \sup_{\beta > 0} \frac{C(\mathcal{N},\beta)}{\beta} & = \sup_{\beta > 0} \frac{1}{\beta} \sup_{\substack{ \left\{ p_X, \rho_{A}(x) \right\}\\ \tr[G \bar \rho_{A}] \le \beta}} I(X;B)_{\rho} \\
    & \le  \sup_{\beta > 0} \sup_{\substack{ \left\{ p_X, \rho_{A}(x) \right\}\\ \tr[G \bar \rho_{A}] \le \beta}} \frac{1}{\tr[ G\bar \rho_A]} I(X;B)_{\rho}\\
    & =  \sup_{ \left\{ p_X, \rho_{A}(x) \right\}} \frac{I(X;B)_{\rho}}{\tr[ G\bar \rho_A]} .
  \end{align}
  Note that the inequality is trivial if for some ensemble $\tr[G \bar \rho_A]=0$ and $I(X;B)_\rho >0$. Now, this is also an achievable rate per unit cost since for any $\left\{ p_X, \rho_A(x) \right\}$, we can achieve a rate per channel use $I(X;B)_\rho$ using cost-constrained Holevo-Schumacher-Westmoreland (HSW) coding~\cite{holevo1998capacity,schumacher1997sending,holevo2003entanglement,H04,holevo2013quantum}. The average cost per channel use is then exactly $\tr[G\bar \rho_A]$, and so we achieve a capacity per unit cost equal to
  \begin{equation}
    \frac{I(X;B)_\rho}{\tr[G\bar \rho_A]}.
  \end{equation}
  This concludes the proof.
\end{proof}

Now, suppose that we have a state $\psi^0$ with zero cost, i.e., $\tr[G \psi^0] =0$. As mentioned above, without loss of generality, $\psi^0$ can be taken pure since otherwise we can spectrally decompose it and conclude that all of its eigenstates must have zero cost since $G \ge 0$. In this special case, the capacity per unit cost is given by the following simple expression: 
\begin{thm}
\label{thm:zero-cost-rel-ent-CC}
  If there is a state $\psi^0$ with zero cost, then the capacity per unit cost of a channel with additive Holevo information  at all cost constraints is
  \begin{equation}
    \label{eq:zeroCCost}
    \textbf{\textit{C}}(\mathcal{N}) = \sup_{\psi \neq \psi^0} \frac{D(\mathcal{N}(\psi) \Vert \mathcal{N}(\psi^0))}{\bra{\psi} G \ket{\psi}},
  \end{equation}
  where $\psi$ is pure.
\end{thm}
Just as was found in~\cite{verdu1990channel}, the expression for the capacity per unit cost is arguably simpler  than that for the capacity cost function given in \cref{eq:capacityCost}. The latter requires an optimization over ensembles on the input space while the former only requires an optimization over the input space itself. 

We now give a proof of Theorem~\ref{thm:zero-cost-rel-ent-CC}.
\begin{proof}
  Without loss of generality, $\psi^0$ is the unique zero-cost state. Otherwise, let $\phi^0 \neq \psi^0$ be a zero-cost state. If $\mathcal{N}(\phi^0) \neq \mathcal{N}(\psi^0)$, the capacity per unit cost is infinite since we can send a binary message with zero cost. If on the other hand $\mathcal{N}(\phi^0) = \mathcal{N}(\psi^0)$, then $\phi^0$ is the same as $\psi^0$ for the purposes of communicating over $\mathcal{N}$. 

  We first prove the direct part. To begin with, we note that the possibility of time-sharing (interpolation between two different protocols) implies that $C(\mathcal{N},\beta)$ is concave in $\beta$. Furthermore, we have a zero-cost state and so $C(\mathcal{N},\beta)/\beta$ is monotone non-increasing on $(0, +\infty)$. We conclude that
  \begin{equation}
    \label{eq:monoZero}
    \textbf{\textit{C}}(\mathcal{N}) = \lim_{\beta \searrow 0} \frac{C(\mathcal{N},\beta)}{\beta}.
  \end{equation}
  Now, let $\beta \in (0,\bra{\psi} G \ket{\psi})$ and consider the following classical-quantum state: 
  \begin{equation}
    \rho_{XA}^\beta = \left( 1- \frac{\beta}{\bra{\psi} G \ket{\psi}} \right) \state{0} \ot \psi^0  + \frac{\beta}{\bra{\psi} G \ket{\psi}} \state{1} \ot \psi ,
  \end{equation}
  where $\psi \neq \psi^0$. By \cref{thm:generalCCost}, we can achieve the following rate per unit cost:
  \begin{equation}
    \frac{I(X;B)_{\rho^\beta}}{\beta}.
  \end{equation}
  Now, recall the following identity, which holds for a classical-quantum state $\rho_{XB} = \sum_{x}^{} p_X(x) \state{x}_X \ot \rho_B^x$~\cite{yuen1993ultimate,holevo2013quantum} where for all $x$, $p_X(x) > 0$:
  \begin{align}
    I(X;B)_\rho & = D(\rho_{XB} \Vert \rho_X \ot \rho_B) \nonumber\\
    & = \sum_{x}^{} p_X(x) D(\rho^x_B \Vert \rho_B) .\label{eq:CQHolevo}
  \end{align}
  This expression is well defined because $\supp(\rho^x_B) \subseteq \supp(\rho_B)$ for all $x$.
  Hence by \cref{eq:CQHolevo} and non-negativity of quantum relative entropy when evaluated on quantum states, we obtain 
  \begin{equation}
    \frac{I(X;B)_{\rho^\beta}}{\beta} \ge \frac{D(\mathcal{N}(\psi )\Vert \rho_B^\beta)}{\bra{\psi} G \ket{\psi}}.
  \end{equation}
  So by the lower semicontinuity of the relative entropy~\cite{holevo2013quantum}, 
  \begin{align}
    \textbf{\textit{C}}(\mathcal{N}) & = \lim_{\beta \searrow 0} \frac{C(\mathcal{N},\beta)}{\beta} \\
    & \ge \lim_{\beta \searrow 0} \frac{I(X;B)_{\rho^\beta}}{\beta} \\
    & \ge \liminf_{\beta \searrow 0} \frac{D(\mathcal{N}(\psi )\Vert \rho_B^\beta)}{\bra{\psi} G \ket{\psi}} \\
    & \ge \frac{D(\mathcal{N}(\psi) \Vert \mathcal{N}(\psi^0))}{\bra{\psi} G \ket{\psi}}.
  \end{align}
  This holds for all $\psi \neq \psi^0$, and so we obtain the direct part
  \begin{equation}
    \textbf{\textit{C}}(\mathcal{N}) \ge \sup_{\psi \neq \psi^0} \frac{D(\mathcal{N}(\psi) \Vert \mathcal{N}(\psi^0))}{\bra{\psi} G \ket{\psi}}.
  \end{equation}
  
  For the converse, we start with 
  \begin{align}
    I(X;B)_\rho & = \inf_{\sigma_B \in \mathcal{S}(\mathcal{H}_B)} D(\rho_{XB} \Vert \rho_X \ot \sigma_B)  \\
    & \le D(\rho_{XB} \Vert \rho_X \ot \mathcal{N}(\psi^0))\\
    & = \sum_{x}^{} p_X(x) D(\mathcal{N}(\rho^x) \Vert \mathcal{N}(\psi^0)).
  \end{align}
  The first equality is a well known identity \cite[Exercise~11.8.2]{wilde2013quantum}.
  Note that if any of the relative entropies are infinite, then the bound is trivial. Therefore, $\forall \beta >0$, \eqref{eq:bunch-1}--\eqref{eq:bunch-last} hold,
  \begin{figure*}[!t]
    \normalsize
  \begin{align}
    \frac{C(\mathcal{N},\beta)}{\beta} &= \frac{1}{\beta} \sup_{\substack{ \left\{ p_X, \rho^x \right\}\\ \tr[G \bar \rho] \le \beta}} I(X;B) \label{eq:bunch-1}\\
    & = \frac{1}{\beta} \sup_{\substack{ \left\{ p_X, \psi^x \right\}\\ \tr[G \bar \psi] \le \beta}} I(X;B) \label{eq:pureHolevo} \\
    & \le \frac{1}{\beta}\sup_{\substack{ \left\{ p_X, \psi^x \right\}\\ \frac{1}{\beta} \tr[G \bar \psi]  \le 1}} \sum_{x}^{} p_X(x) D(\mathcal{N}(\psi^x) \Vert \mathcal{N}(\psi^0)) \\
    & = \frac{1}{\beta}\sup_{\substack{ \left\{ p_X, \psi^x \right\}, \psi^x \neq \psi^0\\ \frac{1}{\beta} \tr[G \bar \psi] \le 1}} \sum_{x}^{} p_X(x) D(\mathcal{N}(\psi^x) \Vert \mathcal{N}(\psi^0))  \\
    & = \sup_{\substack{ \left\{ p_X, \psi^x \right\},\psi^x \neq \psi^0\\ \frac{1}{\beta} \tr[G \bar \psi] \le 1}} \sum_{x}^{} p_X(x) \frac{D(\mathcal{N}(\psi^x) \Vert \mathcal{N}(\psi^0))}{\bra{\psi^x} G \ket{\psi^x}} \frac{\bra{\psi^x} G \ket{\psi^x}}{\beta} \label{eq:divideCost} \\ 
    & \le \sup_{\psi \neq \psi^0} \frac{D(\mathcal{N}(\psi) \Vert \mathcal{N}(\psi^0))}{\bra{\psi} G \ket{\psi}} \sup_{\substack{ \left\{ p_X, \psi^x \right\},\psi^x \neq \psi^0\\ \frac{1}{\beta} \tr[G \bar \psi] \le 1}} \sum_{x}^{} p_X(x)  \frac{\bra{\psi^x} G \ket{\psi^x}}{\beta} \\ 
    & \le \sup_{\psi \neq \psi^0} \frac{D(\mathcal{N}(\psi) \Vert \mathcal{N}(\psi^0))}{\bra{\psi} G \ket{\psi}} ,
    \label{eq:bunch-last}
  \end{align}
    \hrulefill
    \vspace*{4pt}
  \end{figure*}
  \newpage
  \noindent where \cref{eq:pureHolevo} follows since pure state ensembles maximize the Holevo information (even with a cost constraint), and we can divide by $\bra{\psi^x} G \ket{\psi^x}$ in \cref{eq:divideCost} since we assumed $\psi^0$ is the unique zero-cost state.
\end{proof}

\subsection{Pulse-Position-Modulation Scheme for Classical Communication}

We can also directly prove the achievability part of Theorem~\ref{thm:zero-cost-rel-ent-CC} without going through the cost-constrained capacity, as was done in \cite{verdu1990channel} for the classical case. This follows by using a PPM scheme along the following lines. 
\\~\\
\textit{Encoding}: Let $\psi \neq \psi^0$ be a pure state and fix $M , N \in \mathbb{N}$. For a message $m \in [1:M]$, the sender transmits the following length-$MN$ sequence of states: 
\begin{equation}
  \left[(\psi^0)^{\ot N}\right]^{\ot m-1} \ot \psi ^{\ot N} \ot \left[(\psi^0)^{\ot N}\right]^{\ot M-m}.
\end{equation}
That is, the message is encoded in the position of a $\psi$-``pulse'' amidst a baseline of zero-cost states. Note that the cost of each codeword is $N \bra{\psi} G\ket{\psi}$.
\\~\\
\textit{Decoding}: Let $\varepsilon \in (0,1)$. The receiver obtains the state
\begin{equation}
  \label{eq:receivePPM}
  \left[\mathcal{N}(\psi^0)^{\ot N}\right]^{\ot m-1} \ot \mathcal{N}(\psi) ^{\ot N} \ot \left[\mathcal{N}(\psi^0)^{\ot N}\right]^{\ot M-m}.
\end{equation}
Then, the receiver uses a quantum hypothesis test to deduce the position of the pulse. Specifically, he performs $M$ independent binary hypothesis tests with $N$ copies where the null hypothesis is $\mathcal{N}(\psi)$ and the alternative hypothesis is $\mathcal{N}(\psi^0)$. If the receiver obtains a test result of the form \cref{eq:receivePPM} for some $\hat m$, then $\hat m$ is declared. Otherwise an error is declared.
\\~\\
\textit{Error Analysis}: Let $A_{iN}$, for $i \in [1:M]$, denote the POVM of the $i$th hypothesis test,  and let $\alpha_{i N}(A_{iN})$ and $\beta_{iN}(A_{iN})$ denote the Type~I and Type~II errors, respectively. Now, the error probability $p_e$ is independent of the message by symmetry, so we can fix some message index $i$. Furthermore, since each POVM acts on independent size-$N$ blocks, we can apply the classical union bound as follows:
\begin{equation}
  p_e \le \alpha_{iN} + (M-1)\beta_{iN}.
\end{equation}
By \cref{eq:quantumStein}, for $\varepsilon \in (0,1)$,
\begin{equation}
  \lim_{N \to \infty} -\frac{1}{N} \log_2 \beta_{iN}^*(\varepsilon/2) =  D(\mathcal{N}(\psi) \Vert \mathcal{N}(\psi^0)).
\end{equation}
Using the test $A_{iN}$ that achieves $\beta_{iN}^*(\varepsilon/2)$ and given $\delta >0$, for sufficiently large $N$, the probability of error is bounded by
\begin{equation}
  p_e \le \frac{\varepsilon}{2} + (M-1) 2^{-N D(\mathcal{N}(\psi) \Vert \mathcal{N}(\psi^0)) +N \delta}.
\end{equation}
Hence, if
\begin{equation}
  \frac{\log_2 M}{N \bra{\psi}G\ket{\psi}} < \frac{D(\mathcal{N}(\psi) \Vert \mathcal{N}(\psi^0))}{\bra{\psi} G\ket{\psi}} - \frac{2\delta}{\bra{\psi} G\ket{\psi}},
\end{equation}
then $p_e <\varepsilon$ for sufficiently large $N$. Therefore, $\frac{D(\mathcal{N}(\psi) \Vert \mathcal{N}(\psi^0))}{ \bra{\psi} G\ket{\psi}}$ is an achievable rate per unit cost.\\

\section{Entanglement-Assisted Communication}

\label{sec:EA-comm}

We now consider the case of communication with unlimited entanglement assistance. We define an $(n,M,\nu,\varepsilon)$ code in the same way as in the unassisted case, with the exception that the sender and receiver are allowed to share an arbitrary quantum state of arbitrary dimension before communication begins and they can use this resource in the encoding and decoding. The entanglement-assisted capacity cost function $C_\text{EA}(\mathcal{N},\beta)$ is defined similarly but again takes into account the entanglement assistance.

Let $A$ and $A'$ denote quantum systems with isomorphic Hilbert spaces. Define for a bipartite state $\varphi_{AA'}$
\begin{equation}
  \varphi_{AB} \equiv (\id_A \ot \mathcal{N}_{A' \to B})(\varphi_{AA'}).
\end{equation}
We recall the following theorem:
\begin{thm}
  [\cite{holevo2003entanglement}]
  The entanglement-assisted capacity cost function for a quantum channel $\mathcal{N}_{A' \to B}$ is given by
  \begin{equation}
    C_\mathrm{EA}(\mathcal{N},\beta) = \max_{\substack{\varphi_{AA'} \\ \tr[G \varphi_{A'}] \le \beta}} I(A;B)_\varphi,
  \end{equation}
  where $\varphi_{AA'}$ is a pure bipartite state.
\end{thm}
We define the entanglement-assisted capacity per unit cost $\textbf{\textit{C}}_\text{EA}(\mathcal{N})$ in the same manner and obtain an expression for it in the same way as in the unassisted case (i.e., as done in Theorem~\ref{thm:generalCCost}). 
\begin{thm}
  \label{thm:generalEACCost}
  The entanglement-assisted capacity per unit cost for a quantum channel $\mathcal{N}_{A' \to B}$ is given by
  \begin{equation}
    \label{eq:generalEACCost}
    \textbf{\textit{C}}_\mathrm{EA}(\mathcal{N}) = \sup_{\beta >0} \frac{C_\mathrm{EA}(\mathcal{N},\beta)}{\beta} = \sup_{\varphi_{AA'}} \frac{I(A;B)_\varphi}{\tr[G \varphi_{A'}]}.
  \end{equation}
\end{thm}
Now suppose that we have a zero-cost pure state $\psi^0$. 
Similar to the unassisted case, we obtain the expression for the entanglement-assisted capacity per unit cost given in Theorem~\ref{thm:EA-zero-cost-state}.
Note that, like the mutual information, the quantity to be optimized is only a function of the input state $\varphi_{A'}$ and not of the specific purification. Also note that, unlike the unassisted case, $\textbf{\textit{C}}_\mathrm{EA}$ is ostensibly as difficult to calculate as $C_\mathrm{EA}$.

\begin{thm}
\label{thm:EA-zero-cost-state}
  If there is a state $\psi^0$ with zero cost, then the entanglement-assisted capacity per unit cost of a channel $\mathcal{N}_{A' \to B}$ is given by
  \begin{equation}
    \label{eq:zeroEACCost}
    \textbf{\textit{C}}_\mathrm{EA}(\mathcal{N}) = \sup_{\varphi_{AA'}} \textbf{\textit{C}}_{\mathrm{EA},\psi^0}(\mathcal{N}, \varphi),
  \end{equation}
  where
  \begin{multline}
    \textbf{\textit{C}}_{\mathrm{EA},\psi^0}(\mathcal{N}, \varphi)\\
    =
   \begin{cases}
     \frac{D(\varphi_{AB} \Vert \varphi_A \ot \mathcal{N}(\psi^0_{A'}))}{\tr[G \varphi_{A'}]} & D(\varphi_{AB} \Vert \varphi_A \ot \mathcal{N}(\psi^0_{A'})) > 0 \\
     0 & \text{otherwise}
   \end{cases}.
  \end{multline}
\end{thm}
\begin{proof}
  The proof proceeds much like in the unassisted case. Suppose that $\varphi^0 \neq \psi^0$ has zero cost. Now, if $D(\varphi_{AB}^0 \Vert \varphi^0_A \ot \mathcal{N}(\psi^0_{A'})) > 0$, we clearly have infinite $\textbf{\textit{C}}_\mathrm{EA}(\mathcal{N})$ since we can send a distinguishable binary message (using the entangled state $\varphi_{AA'}^0$) with zero cost. Hence it is sufficient to assume that $D(\varphi_{AB}^0 \Vert \varphi^0_A \ot \mathcal{N}(\psi^0_{A'})) = 0$.

  We now prove achievability. By concavity
  of $C_\mathrm{EA}(\mathcal{N},\beta)$ with respect to $\beta$
  and the existence of a zero-cost state,
  \begin{equation}
    \textbf{\textit{C}}_\mathrm{EA}(\mathcal{N}) = \lim_{\beta \searrow 0} \frac{C_\mathrm{EA}(\mathcal{N},\beta)}{\beta}.
  \end{equation}
  Consider some $\varphi_{AA'}$. Since we can trivially achieve zero rate, suppose that $\varphi_{A'}$ has positive cost. Then, define the following state:
  \begin{multline}
    \rho^\beta_{XAA'}  \equiv  \left(1- \frac{\beta}{\tr[G \varphi_{A'}]}\right) \state{0}_X \ot \varphi_A \ot \psi^0_{A'}\\
     +\frac{\beta}{\tr[G\varphi_{A'}]} \state{1}_X \ot \varphi_{AA'} ,
  \end{multline}
  where $\beta \in (0,\tr[G\varphi_{A'}])$.
  By \cref{thm:generalEACCost} and the data-processing inequality for mutual information, we obtain the following entanglement-assisted rate per unit cost for the mixed state $\rho^\beta_{XAB} \equiv (\id_{XA} \ot \mathcal{N}_{A'\to B})(\rho^\beta_{XAA'})$:
  \begin{align}
    \frac{I(XA;B)_{\rho^\beta}}{\beta}. 
  \end{align}
  Now, we can write the mutual information of any classical-quantum state 
  \begin{equation}
    \rho_{XAB} = \sum_{x}^{} p_X(x) \state{x}_X \ot \rho_{AB}^x
  \end{equation}
  as the following convex sum of relative entropies:
  \begin{align}
    & I(XA;B)_\rho \nonumber\\
    & = \tr\left[ \rho_{XAB} \left( \log_2 \rho_{XAB} - \log_2 \rho_{XA} \ot \rho_B \right) \right] \nonumber\\
    & = \sum_{x}^{} p_X(x) \tr_{AB}\left[ \rho_{AB}^x \left( \log_2 \rho_{AB}^x - \log_2 \rho_A^x \ot \rho_B \right) \right] \nonumber\\
    & = \sum_{x}^{} p_X(x) D(\rho_{AB}^x \Vert \rho_A^x \ot \rho_B) \label{eq:EACQ}.
  \end{align}
  Thus, by the non-negativity of quantum relative entropy when evaluated on quantum states,
  \begin{align}
    \frac{I(XA;B)_{\rho^\beta}}{\beta} & \ge \frac{D\!\left( \varphi_{AB} \Vert \varphi_A \ot \rho_B^\beta\right)}{\tr[ G \varphi_{A'}]}.
  \end{align}
  Hence, again by the lower semicontinuity of the relative entropy, 
  \begin{align}
    \textbf{\textit{C}}_\mathrm{EA}(\mathcal{N}) & = \lim_{\beta \searrow 0} \frac{C_\mathrm{EA}(\mathcal{N},\beta)}{\beta} \\
    & \ge \lim_{\beta \searrow 0} \frac{I(XA;B)_{\rho^\beta}}{\beta} \\
    & \ge \liminf_{\beta \searrow 0} \frac{D\!\left( \varphi_{AB} \Vert \varphi_A \ot \rho_B^\beta\right)}{\tr[ G \varphi_{A'}]} \\
    & \ge \frac{D( \varphi_{AB} \Vert \varphi_A \ot \mathcal{N}(\psi^0_{A'}))}{\tr[ G \varphi_{A'}]} .
  \end{align}

  For the converse, we have for any pure input state $\varphi_{AA'}$,
  \begin{align}
    I(A;B)_\varphi & = \inf_{\sigma_B} D(\varphi_{AB} \Vert \varphi_A \ot \sigma_B)\\
    & \le D( \varphi_{AB} \Vert \varphi_A \ot \mathcal{N}(\psi^0_{A'})).
  \end{align}
  Note again that if the relative entropy is infinite, then the bound is trivial. Hence,
  \begin{align}
    \frac{C_\mathrm{EA}(\mathcal{N},\beta)}{\beta} & \le \sup_{\substack{ \varphi_{AA'} \\ \tr[G \varphi_{A'}] \le \beta} } \frac{D( \varphi_{AB} \Vert \varphi_A \ot \mathcal{N}(\psi^0_{A'}) )}{\beta}. \label{eq:EAconverse2}
  \end{align}
  Now, we assumed that for any zero-cost state $\varphi^0$, $D( \varphi_{AB} \Vert \varphi^0_A \ot \mathcal{N}(\psi^0_{A'}) ) =0$. Thus we can take the supremum over non-zero cost states. If there are not any, then \cref{eq:EAconverse2} implies that the upper bound is 0, which would conclude the converse. Otherwise, we can argue
  \begin{align}
    \frac{C_\mathrm{EA}(\mathcal{N},\beta)}{\beta} & \le \sup_{\substack{ \varphi_{AA'} \\ \frac{1}{\beta} \tr[G \varphi_{A'}] \le 1 } } \frac{D( \varphi_{AB} \Vert \varphi_A \ot \mathcal{N}(\psi^0_{A'}) )}{\tr[G \varphi_{A'}]} \frac{\tr[G \varphi_{A'}]}{\beta}\\
    & \le \sup_{\varphi_{AA'} }\frac{D( \varphi_{AB} \Vert \varphi_A \ot \mathcal{N}(\psi^0_{A'}) )}{\tr[G \varphi_{A'}]} .
  \end{align}
  This concludes the proof.
\end{proof}

\subsection{Pulse-Position-Modulation Scheme for Entanglement-Assisted Classical Communication}

We propose a PPM scheme that achieves the rate given in \cref{eq:zeroEACCost}, thereby providing an alternative proof of the direct part of Theorem~\ref{thm:EA-zero-cost-state}. This will be much like the scheme in the unassisted case except with the greater discriminatory power that entanglement assistance provides.
\\~\\
\textit{Encoding}: Let $\varphi_{A'}$ be a positive-cost state and fix $M , N\in \mathbb{N}$. The sender and receiver share $M N$ copies of a pure state $\varphi_{AA'}$, where $A'$ is at the sender and $A$ is at the receiver, where we have $N$ copies for each message in $[1:M]$. Hence the overall shared state is
\begin{equation}
   \bigotimes_{i=1}^{M}\left( \varphi_{A_i A_i'}\right)^{\ot N}.
\end{equation}
For a message $m \in [1:M]$, the sender transmits a $\varphi_{A'}$-pulse amidst a zero-cost state baseline by using the following sequence of states:
\begin{equation}
   \bigotimes_{i=1}^{m-1}\left( \psi^0_{A_i'} \right)^{\ot N} \ot \left(\varphi_{A'_{m}}\right)^{\ot N} \ot  \bigotimes_{j=m+1}^{M} \left(\psi^0_{A_j'} \right)^{\otimes N}.
\end{equation}
That is, the sender transmits the zero-cost state, but at every $m$th block of length $N$, he sends his shares of the corresponding copies of $\varphi_{AA'}$. Note that the cost of each codeword is $N\tr[G \varphi_{A'}]$.
\\~\\
\textit{Decoding}: Let $\varepsilon \in (0,1)$. Now, since $\varphi_{AA'}$ purifies $\varphi_{A'}$, whenever the sender transmits $\psi^0$, the receiver obtains a product state $\varphi_A \ot \mathcal{N}(\psi^0)$. Hence, the receiver now has the state
\begin{multline}
  \label{eq:EAreceivePPM}
   \bigotimes_{i=1}^{m-1} \left(\varphi_{A_i} \ot\mathcal{N}( \psi^0_{A_i'} )\right)^{\ot N} \ot (\id_A \ot \mathcal{N}) (\varphi_{A_mA_m'})^{\ot N} \\
   \ot \bigotimes_{j=m+1}^{M} \left( \varphi_{A_j} \ot\mathcal{N}(\psi^0_{A_j'})\right)^{\ot N}.
\end{multline}
Then, the receiver uses quantum hypothesis testing along with his shares of the entangled states to deduce the position of the pulse. He performs $M$ independent binary hypothesis tests with $N$ copies where the null hypothesis is $\varphi_{AB} \equiv (\id_A \ot \mathcal{N}_{A' \to B})(\varphi_{AA'})$ and the alternative hypothesis is $\varphi_A \ot \mathcal{N}(\psi^0_{A'})$. If the receiver obtains a test result of the form \cref{eq:EAreceivePPM} for some $\hat m$, then $\hat m$ is declared. Otherwise an error is declared.
\\~\\
\textit{Error Analysis}: The error analysis follows in exactly the same way as the unassisted case. We conclude that we can obtain vanishing error in transmission provided that, for some $\delta > 0$,
\begin{equation}
  \frac{\log_2 M}{N \tr[G \varphi_{A'}]} < \frac{D(\varphi_{AB} \Vert \varphi_A \ot \mathcal{N}(\psi^0_{A'}))}{\tr[G \varphi_{A'}]} - \frac{\delta}{\tr[G \varphi_{A'}]}.
\end{equation}
Hence we achieve the entanglement-assisted rate per unit cost $\frac{D(\varphi_{AB} \Vert \varphi_A \ot \mathcal{N}(\psi^0_{A'}))}{\tr[G \varphi_{A'}]}$. 

Similar to the position-based coding scheme for entanglement-assisted classical communication~\cite{anshu2017one}, this scheme does not consume all the entanglement needed to implement the encoding. This follows from the gentle-measurement lemma \cite{itit1999winter,thesis1999winter}: the entangled states that were not transmitted but measured by the decoder will only be negligibly disturbed, given that the decoding measurement succeeds with high probability. Now,
the natural measure of rate of entanglement consumption in this setting is
 entanglement consumed per unit cost,  and in  this scheme it can be expressed in terms of the entanglement entropy of $\varphi_{AA'}$ as follows:
\begin{equation}
  \frac{N S(A)_\varphi}{N \tr[G \varphi_{A'}]} = \frac{S(A)_\varphi}{ \tr[ G \varphi_{A'}]}.
\end{equation}

\section{Private Communication}

\label{sec:priv-comm}

We now consider private communication over a quantum channel. This was first studied in \cite{devetak2005private,1050633} when there is no cost constraint and recently in \cite{wilde2016energy} when there is a cost constraint.  Given a noisy channel $\mathcal{N}_{A \to B}$, let $\mathcal{U}_{A \to BE}$ denote an isometric channel extending it and let $\mathcal{N}^c_{A \to E} = \tr_B \circ \ \mathcal{U}_{A \to BE}$ denote the induced complementary channel. 
A channel $\mathcal{N}_{A \to B}$ is degradable if there exists a degrading channel $\mathcal{D}_{B \to E}$ such that $\mathcal{N}^c_{A \to E} = \mathcal{D}_{B \to E} \circ \mathcal{N}_{A \to B}$ \cite{cmp2005dev}.

The formulation here is based on~\cite{el2013secrecy}, but note that here we use a definition of a private code with the privacy based on trace distance~\cite{KRBM07,wilde2013quantum,wilde2016energy}. Namely, we define an $(n, M, \nu, \varepsilon, \zeta)$ private code as having blocklength $n\in \mathbb{N}$, number of messages $M\in \mathbb{N}$, total cost at most $\nu\in \mathbb{R}_{> 0}$, and probability of error of the receiver's decoding at most $\varepsilon\in [0,1]$. The quantity $\zeta\in[0,1]$ bounds the privacy error for each message: for each message, we demand that the eavesdropper's state is approximately independent of the message. Specifically, for all $m \in [1:M]$, we require that
\begin{equation}
  \frac{1}{2} \norm{ \left( \mathcal{N}^c_{A \to E} \right)^{\ot n} (\rho_{A^n}(m)) - \sigma_{E^n}}_1 \leq \zeta,
\end{equation}
where $\rho_{A^n}(m)$ are the codewords and  $\sigma_{E^n}$ is some fixed state independent of $m$.

We can now establish some definitions.
\begin{dfn}
  [\cite{wilde2016energy}]
  Given $\beta >0$, $R_p$ is an achievable private communication rate with average cost not exceeding $\beta$ if for all $\varepsilon, \zeta \in (0,1)$ and
  $\delta>0$, $\exists \, n_0$ such that $\forall n \ge n_0$, there is an $(n, M, n\beta, \varepsilon, \zeta)$ code for which
  \begin{equation}
    \frac{\log_2 M}{n} > R_p - \delta.
  \end{equation}
  The supremum of all achievable rates with average cost not exceeding $\beta$ as a function of $\beta$ is the private capacity cost function $P(\mathcal{N},\beta)$.
\end{dfn}
We recall the formula for $P(\mathcal{N},\beta)$ when $\mathcal{N}$ is a degradable channel:
\begin{thm}
  [\cite{wilde2016energy}]
  The private capacity cost function for a degradable channel $\mathcal{N}_{A \to B}$ is given by 
  \begin{equation}
    P(\mathcal{N},\beta) = \sup_{\substack{ \left\{ p_X, \psi_A(x) \right\} \\ \tr[G \bar \psi_A] \le \beta }} I(X;B)_\rho - I(X;E)_\rho,
  \end{equation}
  where each state $\psi_A(x)$ is pure,
  \begin{equation}
    \rho_{XBE} = \sum_{x}^{} \state{x}_X \ot \mathcal{U}_{A \to BE} (\psi_A(x)),
  \end{equation}
  and $\bar \psi_A = \sum_x p_X(x) \psi_A(x)$ is the average input state.
\end{thm}

Now we give the definition for the private capacity per unit cost.
\begin{dfn}
  $\textbf{\textit{R}}_p$ is an achievable private communication rate per unit cost if for all $\varepsilon, \zeta \in (0,1)$ and
  $\delta>0$, there $\exists \, \nu_0 > 0$ such that $\forall \nu \ge \nu_0$, there is an $(n, M, \nu, \varepsilon, \zeta)$ code for which
  \begin{equation}
    \log_2 M > \nu(\textbf{\textit{R}}_p -\delta).
  \end{equation}
  The private capacity per unit cost is equal to the supremum of all achievable private communication rates per unit cost, denoted by $\textbf{\textit{P}}(\mathcal{N})$.
\end{dfn}

We can obtain an expression for the private capacity per unit cost via the private capacity cost function.  A proof of this follows from the achievability and converse of the cost-constrained private capacity per channel use \cite{wilde2016energy}, just as in the proof of \cref{thm:generalCCost}. 
\begin{thm}
  \label{thm:generalPCost}
  The private capacity per unit cost of a degradable channel $\mathcal{N}_{A \to B}$ is given by
  \begin{equation}
    \textbf{\textit{P}}(\mathcal{N}) = \sup_{\beta > 0} \frac{P(\mathcal{N},\beta)}{\beta} = \sup_{\{p_X, \psi_A(x)\}} \frac{I(X;B)_\rho - I(X;E)_\rho}{\tr[ G \bar \psi_A]},
  \end{equation}
  where $\bar \psi_A$ is the average input state.
\end{thm}

Now again suppose that we have a zero-cost state $\psi^0$. We then obtain the following simpler expression for the private capacity per unit cost.
\begin{thm}
  \label{thm:zeroPrivate}
  If there is a state $\psi^0$ with zero cost, then the private capacity per unit cost of a degradable channel $\mathcal{N}_{A \to B}$ is given by
  \begin{equation}
    \textbf{\textit{P}}(\mathcal{N}) = \sup_{\psi} \textbf{\textit{P}}_{\psi^0}(\mathcal{N},\psi),
  \end{equation}
  where $\psi$ is pure,
  \begin{align}
    \textbf{\textit{P}}_{\psi^0}(\mathcal{N},\psi) & \equiv 
    \begin{cases}
    \frac{N_\mathcal{N}(\psi,\psi^0)}{\bra{\psi} G \ket{\psi}} & N_\mathcal{N}(\psi,\psi^0) > 0\\
      0 & \text{otherwise}
    \end{cases}
    ,\\
    N_\mathcal{N}(\psi,\psi^0) & \equiv D(\mathcal{N}(\psi) \Vert \mathcal{N}(\psi^0)) - D(\mathcal{N}^c(\psi) \Vert \mathcal{N}^c(\psi^0)),
  \end{align}
  and $\mathcal{N}^c_{A \to E}$ denotes the complementary channel of $\mathcal{N}$ corresponding to an isometric channel $\mathcal{U}_{A \to BE}$ extending $\mathcal{N}$. 
\end{thm}
\begin{proof}
  Suppose that the state $\varphi^0$ has zero cost. If $N_\mathcal{N}(\varphi^0,\psi^0) \neq 0$, we have a zero-cost binary alphabet over which we can form ensembles for which $I(X;B) - I(X;E) > 0$. Hence, by \cref{thm:generalPCost}, we can attain infinite private capacity per unit cost. Thus, it suffices to assume $N_\mathcal{N}(\varphi^0, \psi^0) = 0$. 

  Now, once again by concavity and the existence of a zero-cost state,
  \begin{equation}
    \textbf{\textit{P}}(\mathcal{N}) = \lim_{\beta \searrow 0} \frac{P(\mathcal{N},\beta)}{\beta}.
  \end{equation}
  Let $\beta \in (0,\bra{\psi} G \ket{\psi})$, and let $\psi$ be a pure state. We can assume that $\psi$ has positive cost since it is trivial to attain zero rate. Consider the following classical-quantum state:
  \begin{multline}
     \rho_{XBE}^\beta =  \left(1 - \frac{\beta}{\bra{\psi} G \ket{\psi}}\right) \state{0}_X \ot \mathcal{U}_{A \to BE}(\psi^0_A)  \\
     +
    \frac{\beta}{\bra{\psi} G \ket{\psi}} \state{1}_X \ot \mathcal{U}_{A \to BE}(\psi_A) .
  \end{multline}
By a similar argument as in the unassisted case applied to each relative entropy in $N_\mathcal{N}$, in the limit $\beta \searrow 0$, this ensemble achieves the desired quantity:
  \begin{equation}
    \textbf{\textit{P}}(\mathcal{N}) \ge \lim_{\beta \searrow 0} \frac{I(X;B)_{\rho^\beta} - I(X;E)_{\rho^\beta}}{\beta} \ge \frac{N_\mathcal{N}(\psi,\psi^0)}{\bra{\psi}G\ket{\psi}}.
  \end{equation}
  In arriving at the above result, we need to make use of the lower semi-continuity of the private information as a function of the input ensemble. This is proven for bounded cost ensembles in certain settings in Corollary 3 of~\cite{shirokov2016lower} and in particular applies to our  case here. 

  For the converse, we note that for any ensemble $\left\{ p_X, \rho^x \right\}$,
  \begin{align}
    & I(X;B)_\rho - I(X;E)_\rho \nonumber \\
    & = \sum_{x}^{}p_X(x) N_\mathcal{N}(\rho^x, \bar \rho) \\
    & = \sum_{x}^{}p_X(x) N_\mathcal{N}(\rho^x, \psi^0) - N_\mathcal{N}(\bar \rho, \psi^0) \\
    & \le \sum_{x}^{}p_X(x) N_\mathcal{N}(\rho^x, \psi^0).
  \end{align}
The inequality follows since by degradability of $\mathcal{N}$ and monotonicity of relative entropy \cite{Lindblad1975}, $N_\mathcal{N}(\rho, \sigma) \ge 0$ for all states $\rho$ and $\sigma$. If any of the $N_\mathcal{N}$ quantities are infinite, the bound is trivial. 

  We can then argue for all $\beta > 0$,
  \begin{align}
    \frac{P(\mathcal{N},\beta)}{\beta} 
    & = \frac{1}{\beta} \sup_{\substack{ \left\{ p_X, \psi^x \right\} \\ \tr[G \bar \psi] \le \beta}} I(X;B) - I(X;E)\\
    & \le \frac{1}{\beta} \sup_{\substack{ \left\{ p_X, \psi^x \right\} \\ \tr[G \bar \psi] \le \beta}} \sum_{x}^{}p_X(x) N_\mathcal{N}(\psi^x, \psi^0) \label{eq:PConv1}.
  \end{align}
  Now, just as in the entanglement-assisted case, we can restrict the supremum to be taken over positive-cost states $\psi^x$:
  \begin{align}
    \frac{P(\mathcal{N},\beta)}{\beta }& \le \frac{1}{\beta} \sup_{\substack{ \left\{ p_X, \psi^x \right\} \\ \frac{1}{\beta} \tr[G \bar \psi] \le 1}} \sum_{x}^{}p_X(x) N_\mathcal{N}(\psi^x, \psi^0) \\
    & = \sup_{\substack{ \left\{ p_X, \psi^x \right\} \\ \frac{1}{\beta} \tr[G \bar \psi] \le 1}} \sum_{x}^{}p_X(x) \frac{N_\mathcal{N}(\psi^x, \psi^0)}{\bra{\psi^x} G \ket{\psi^x}} \frac{\bra{\psi^x} G \ket{\psi^x}}{\beta}  \\
    & \le \sup_{\psi} \frac{N_\mathcal{N}(\psi, \psi^0)}{\bra{\psi} G \ket{\psi}}.
  \end{align}
  This concludes the proof.
\end{proof}

\subsection{Pulse-Position-Modulation Scheme for Private Communication}

We now give an alternative proof of the achievability part of Theorem~\ref{thm:zeroPrivate} via a PPM scheme that achieves the private capacity per unit cost for a degradable channel $\mathcal{N}_{A \to B}$. 
\\~\\
 \textit{Codebook}: As discussed above, without loss of generality, we can  restrict the discussion to positive-cost pure states $\psi$ such that $N_\mathcal{N}(\psi, \psi^0) >0$. Let $\psi$ be such a state. Then, fix $M, L, N \in \mathbb{N} $. We have $ML$ codewords labeled by $(m,l)$ where $m \in [1:M]$ and $l \in [1:L]$. For codeword $(m,l)$, the corresponding input quantum state is
\begin{multline}
   \left[ \left( (\psi^0)^{ \ot N} \right)^{\ot L} \right]^{\ot m-1} \\
   \ot \left[ \left( (\psi^0)^{ \ot N} \right)^{\ot l-1} \ot \psi^{\ot N} \ot \left( (\psi^0)^{ \ot N} \right)^{\ot L-l} \right] \\
   \ot \left[ \left( (\psi^0)^{ \ot N} \right)^{\ot L} \right]^{\ot M-m}.
\end{multline}
This can be understood as a $\psi$-pulse amidst a baseline of $\psi^0$ states, which is itself a pulse amidst a $(\psi^0)^{\ot L}$ baseline. Note that the cost of this codeword is $N \bra{\psi} G \ket{\psi}$.
\\~\\
\textit{Encoding}: The sender transmits the message $m$ to the receiver and uses $l$ to obscure the message on the eavesdropper's side. Given message $m \in [1:M]$, he uniformly chooses at random $l \in [1:L]$ and transmits $N$ times the quantum state corresponding to $(m, l)$.
\\~\\
\textit{Decoding}: The receiver performs $ML$ binary quantum hypothesis tests, using $N$ copies for each test, in order to determine the pulse position. Again this can be done with vanishing error provided that, for some $\delta >0$,
\begin{equation}
  \frac{\log_2 (ML)}{N}  < D(\mathcal{N}(\psi) \Vert \mathcal{N}(\psi^0))  - \delta.
\end{equation}
\\
\textit{Privacy}: Given the randomization over $l$, the state that the eavesdropper obtains is 
\begin{align}
  & \left[ \left( \mathcal{N}^c(\psi^0)^{\ot N} \right)^{\ot L} \right]^{\ot m-1} \ot \xi_L \ot \left[ \left( \mathcal{N}^c(\psi^0)^{\ot N} \right)^{\ot L} \right]^{\ot M-m}, \label{eq:privateEveState}
\end{align}
where
\begin{multline}
   \xi_L \equiv \\
   \frac{1}{L} \sum_{l=1}^L \left( \mathcal{N}^c(\psi^0)^{\ot N} \right)^{\ot l-1} \ot \mathcal{N}^c(\psi)^{\ot N} \ot \left( \mathcal{N}^c(\psi^0)^{\ot N} \right)^{\ot L-l}.
\end{multline}
Now, by a corollary to the convex-split lemma~\cite{anshu2017one}, the state in \cref{eq:privateEveState} is approximately $\mathcal{N}^c(\psi^0)^{\ot NLM}$ if $L$ is chosen large enough. More precisely, given $\delta', \varepsilon > 0$,
\begin{align}
   & \frac{1}{2} \Big\Vert \mathcal{N}^c(\psi^0)^{\ot NL(m-1)} \ot  \xi_L \ot \mathcal{N}^c(\psi^0)^{\ot NL(M-m)} \nonumber\\
   & \quad - \mathcal{N}^c(\psi^0)^{\ot NLM} \Big\Vert_1 \nonumber\\
   & = \frac{1}{2} \norm{ \xi_L - \left( \mathcal{N}^c(\psi^0)^{\ot N}  \right)^{\ot L} }_1 \nonumber \\
   & \le 2 \varepsilon +\delta'   \label{eq:convexSplit}
\end{align}
if 
\begin{equation}
  \label{eq:randomSeedBound}
  L >  2^{D_{\max}^\varepsilon(\mathcal{N}^c(\psi)^{\ot N} \Vert \mathcal{N}^c(\psi^0)^{\ot N}) } \delta'^{-2},
\end{equation}
where $D_{\max}^\varepsilon$ denotes the smooth max-relative entropy \cite{D09}\footnote{The original definition differs slightly from the definition in~\cite{anshu2017one}, which is the one we use here.}:
\begin{equation}
  \label{eq:defnSmoothMax}
  D_{\max}^\varepsilon( \rho \Vert \sigma) \equiv \inf_{\tilde \rho \in B^\varepsilon(\rho)} D_{\max}( \tilde \rho \Vert \sigma).
\end{equation}
Here, $B^\varepsilon(\rho)$ denotes the $\varepsilon$-ball around $\rho$:
\begin{equation}
  B^\varepsilon(\rho) \equiv \left\{ \tilde \rho \ge 0 : \sqrt{1 - F^2(\tilde \rho, \rho)} \le \varepsilon , \tr[\tilde \rho] = 1 \right\},
\end{equation}
where $F$ is the quantum fidelity~\cite{U76}, and
\begin{equation}
  D_{\max}(\rho \Vert \sigma) \equiv \log \inf\left\{ \lambda \ge 0 : \rho \le \lambda  \sigma \right\}.
\end{equation}
Hence, using the quantum asymptotic equipartition property~\cite{tomamichel2015quantum} for smooth max-relative entropy, for small enough $\varepsilon$, all $\delta'' >0$, and sufficiently large $N$, \cref{eq:randomSeedBound} is satisfied if 
\begin{equation}
  L > 2^{N \left(D(\mathcal{N}^c(\psi) \Vert \mathcal{N}^c(\psi^0)) + \delta''\right)} \delta'^{-2}.
\end{equation}
Taking the logarithm on both sides and dividing by $N$, we obtain
\begin{equation}
  \frac{\log_2 L}{N} > D(\mathcal{N}^c(\psi) \Vert \mathcal{N}^c(\psi^0)) + \delta'' - \frac{2 \log_2 \delta'}{N}.
\end{equation}
For large enough $N$, the condition becomes 
\begin{equation}
  \frac{\log_2 L}{N} > D(\mathcal{N}^c(\psi) \Vert \mathcal{N}^c(\psi^0)) + \delta'''
\end{equation}
for some $\delta''' >0$. 

We conclude that we can attain arbitrarily low decoding and privacy error if
\begin{equation}
  \frac{\log_2 (ML)}{N}  < D(\mathcal{N}(\psi) \Vert \mathcal{N}(\psi^0))  - \delta
\end{equation}
and
\begin{equation}
  \frac{\log_2 L}{N} > D(\mathcal{N}^c(\psi) \Vert \mathcal{N}^c(\psi^0)) + \delta'''.
\end{equation}
Combining the two inequalities, we obtain
\begin{equation}
  \frac{\log_2 M}{N \bra{\psi} G \ket{\psi}} < \frac{N_\mathcal{N}(\psi, \psi^0)}{\bra{\psi} G \ket{\psi} } - \frac{\delta + \delta'''}{\bra{\psi} G \ket{\psi}}.
\end{equation}
Hence, we can achieve a private rate per unit cost of $\frac{N_\mathcal{N}(\psi, \psi^0)}{\bra{\psi} G \ket{\psi}}$.

Note that the above protocol does not use the fact that the channel is degradable. Hence, this gives an achievability result for non-degradable quantum channels as well. Furthermore, we can achieve this with a mixed state $\rho$ instead of a pure state $\psi$, which might be necessary for non-degradable channels. We conclude for general quantum channels $\mathcal{N}$,
\begin{equation}
  \textbf{\textit{P}}(\mathcal{N}) \ge \sup_{\rho \in \mathcal{S}(\mathcal{H}_A)} \textbf{\textit{P}}_{\psi^0}(\mathcal{N}, \rho).
\end{equation}

\section{Quantum Communication}

\label{sec:quant-comm}

We now formulate the quantum capacity per unit cost. Quantum capacity was first studied in \cite{PhysRevA.54.2629,PhysRevA.55.1613,BKN98,BNS98,capacity2002shor,devetak2005private} when there is no cost constraint and recently in \cite{wilde2016energy} when there is a cost constraint. Let $\mathcal{N}_{A \to B}$ be a quantum channel, and let $n,Q\in\mathbb{N}$, $\nu\in\mathbb{R}_{> 0}$, and $\varepsilon\in[0,1]$. An $(n,Q,\nu,\varepsilon)$ quantum code has blocklength $n$, dimension $Q$ for the total input space, total cost at most $\nu$, and quantum decoding error at most $\varepsilon$. In more detail, an $(n,Q,\nu, \varepsilon)$ code for quantum communication  consists of an encoding channel $\mathcal{E}^{n}:\mathcal{S}%
(\mathcal{H}_{S})\rightarrow\mathcal{S}(\mathcal{H}_{A}^{\otimes n})$ and a
decoding channel $\mathcal{D}^{n}:\mathcal{S}(\mathcal{H}_{B}^{\otimes
n})\rightarrow\mathcal{S}(\mathcal{H}_{S})$, where $\dim(\mathcal{H}_{S})=Q$.
The cost constraint imposes the following bound 
on all states resulting from the output of the encoding
channel $\mathcal{E}^{n}$:
\begin{equation}
\operatorname{Tr}\left\{  G_{n}\mathcal{E}^{n}(\rho_{S})\right\}
\leq \nu, \label{eq:q-code-unif-energy-const}%
\end{equation}
where $\rho_{S}\in\mathcal{S}(\mathcal{H}_{S})$
and $G_n$ is defined in \cref{eq:cost-observable}.  
Finally, we have the error bounded by $\varepsilon$, in the sense that for all pure states $\phi_{RS}\in\mathcal{S}(\mathcal{H}_{R}%
\otimes\mathcal{H}_{S})$, where $\mathcal{H}_{R}$ is isomorphic
to$~\mathcal{H}_{S}$, the following trace distance bound holds:
\begin{equation}
  \frac{1}{2} \norm{ \phi_{RS} - (\operatorname{id}_{R}\otimes\lbrack\mathcal{D}^{n}\circ
\mathcal{N}^{\otimes n}\circ\mathcal{E}^{n}])(\phi_{RS})}_1 \le \varepsilon.
\label{eq:q-code-fidelity}%
\end{equation}

\begin{dfn}
  Given $\beta >0$, $R_q$ is an achievable quantum communication rate with average cost not exceeding $\beta$ if for all $\varepsilon \in(0,1)$ and $\delta >0 $, $\exists \, n_0$ such that for $n \ge n_0$, $\exists$ an $(n, Q, n \beta, \varepsilon)$ code for which
  \begin{equation}
    \frac{\log_2 Q}{n} > R_q - \delta.
  \end{equation}
 The quantum capacity cost function $Q(\mathcal{N},\beta)$ is equal to the supremum of all achievable quantum communication rates with average cost not exceeding $\beta$.
\end{dfn}
Building on \cite{cmp2005dev}, the expression for $Q(\mathcal{N},\beta)$ when $\mathcal{N}$ is degradable was shown in \cite{wilde2016energy} to be the following.
\begin{thm}
  [\cite{wilde2016energy}]
  The quantum capacity cost function for a degradable channel $\mathcal{N}_{A \to B}$ is given by
  \begin{equation}
    Q(\mathcal{N},\beta) = \sup_{\substack{ \varphi_{RA} \\ \tr[G \varphi_{A}] \le \beta}} I(R \rangle B)_\varphi,
  \end{equation}
  where $I(R \rangle B) \equiv S(B) - S(RB)$ is the coherent information \cite{PhysRevA.54.2629} and
  \begin{equation}
    \varphi_{RB} \equiv (\id_R \ot \mathcal{N}_{A \to B})(\varphi_{RA}).
  \end{equation}
\end{thm}
We then define the quantum capacity per unit cost.
\begin{dfn}
  $\textbf{\textit{R}}_q$ is an achievable quantum communication rate per unit cost if for any $\delta, \varepsilon > 0$, $\exists \, \nu_0 > 0$ such that for $\nu \ge  \nu_0$ there is an $(n, Q, \nu, \varepsilon)$ code for which
  \begin{equation}
    \log_2 Q > \nu(\textbf{\textit{R}}_q - \delta).
  \end{equation}
  The quantum capacity per unit cost is then defined to be the supremum of all achievable quantum communication rates per unit cost and is denoted as $\textbf{\textit{Q}}(\mathcal{N})$.
\end{dfn}
Using the achievability and converse proofs for $Q(\mathcal{N},\beta)$, and reasoning similar to that in the proof of Theorem~\ref{thm:generalCCost}, we obtain the following expressions for $\textbf{\textit{Q}}(\mathcal{N})$.
\begin{thm}
  \label{thm:generalQCost}
  The quantum capacity per unit cost for a degradable channel $\mathcal{N}_{A \to B}$ is given by
  \begin{equation}
    \textbf{\textit{Q}}(\mathcal{N}) = \sup_{\beta >0} \frac{Q(\mathcal{N},\beta)}{\beta} = \sup_{\varphi_{RA}} \frac{I(R \rangle B)_\varphi}{\tr[G \varphi_{A}]}.
  \end{equation}
\end{thm}

However, note that since $Q(\mathcal{N},\beta) = P(\mathcal{N},\beta)$~\cite{S08,wilde2016energy} for degradable channels, in this case the quantum capacity per unit cost is equal to the private capacity per unit cost. In particular, if a zero-cost state exists, then the quantum capacity per unit cost of a degradable channel is given by the expression in \cref{thm:zeroPrivate}.

\subsection{Pulse-Position-Modulation Scheme for Quantum Communication}

We propose a PPM scheme for achieving the quantum capacity per unit cost for a
degradable channel $\mathcal{N}_{A\rightarrow B}$. We do this by operating the PPM scheme for private
communication in a coherent fashion analogous to that
of~\cite{devetak2005private}. Since the approach is so similar (yet tailored to a PPM coding scheme), we merely sketch the proof for
simplicity and point to \cite{devetak2005private} for more details
(see also \cite{DW04} in this context). The task
we consider is entanglement generation, in which the goal is to establish a
maximally entangled state between the sender and receiver. To generalize this to arbitrary quantum states, we again point to~\cite{devetak2005private}.
Let $\mathcal{U}%
_{A\rightarrow BE}$ be an isometric channel extending $\mathcal{N}%
_{A\rightarrow B}$ and let $U_{A\rightarrow BE}$ denote  the corresponding isometry. Let $\psi^{0}$ be a zero-cost pure state and $\psi$ a positive-cost pure state for which $N_{\mathcal{N}}(\psi,\psi^{0})>0$. 

We first consider the case in which $\psi$ is orthogonal to $\psi^0$. 
~\newline\textit{Codebook}: 
As in the private PPM
scheme, we fix $M,L,N\in\mathbb{N}$. For each ordered pair $(m,l)$, where
$m\in\lbrack1:M]$ and $l\in\lbrack1:L]$, consider the following pure quantum state:%
\begin{multline}
  \psi^{m,l} \equiv \left[  \left(  (\psi^{0})^{\ot N}\right)  ^{\ot L}\right]
^{\ot m-1} \\
   \ot\left[  \left(  (\psi^{0})^{\ot N}\right)  ^{\ot l-1}\ot\psi^{\ot
N}\ot\left(  (\psi^{0})^{\ot N}\right)  ^{\ot L-l}\right]  \\
   \ot\left[  \left(
  (\psi^{0})^{\ot N}\right)  ^{\ot L}\right]  ^{\ot M-m} .
  \label{eq:quantumPulse}
\end{multline}
Note that since $\psi$ is orthogonal to $\psi^0$, $\psi^{m,l}$ are orthogonal for different $m,l$. Observe also that since $G \ket{\psi^0} = 0$,
\begin{multline}
   G_K \ket{ \psi^{m,l}}_{A^K} \\
   = \sum_{j= n(m,l)}^{n(m,l) + N-1} \left( I^{\ot j -1} \ot G \ot I^{\ot K - j}\right) \ket{ \psi^{m,l}}_{A^K},
  \label{eq:costQuantumPulse}
\end{multline}
where $K \equiv NML$ and $n(m,l) \equiv NL(m-1) +N(l-1)+1$ is the position of the first $\psi$ state in $\psi^{m,l}$.
\newline~\newline\textit{Encoding}: Let $\ket{\Phi}_{R\hat{A}}$ denote the maximally entangled 
state to be established with the receiver, where the dimension of $\hat A$, the
system the sender is to transmit, is at most $M$. We can decompose the state
with respect to some orthonormal bases $\left\{  \ket{m}_{R}\right\}  $ and
$\{\ket{m}_{\hat{A}}\}$:
\begin{equation}
\ket{\Phi}_{R\hat{A}}=\frac{1}{\sqrt{M}}\sum_{m}\ket{m}_{R}\ket{m}_{\hat{A}}.
\end{equation}
Depending on the value of $m$, the sender coherently prepares a uniform
superposition of $\ket{\psi^{m,l}}$ over the $l$ variable. That is, the mapping is given by
\begin{equation}
\ket{m}_{\hat{A}}\mapsto\frac{1}{\sqrt{L}}\sum_{l=1}^{L}\ket{\psi^{m,l}}_{A^K}.
\end{equation}
Since $\psi^{m,l}$ are orthogonal for different $m,l$,
the above mapping is an isometry. The overall state after the
encoding is%
\begin{equation}
\ket{\Psi}_{RA^{K}}=\frac{1}{\sqrt{ML}}\sum_{m,l}\ket{m}_{R}%
\ket{\psi^{m,l}}_{A^{K}}.
\end{equation}
Note that the reduced state on $A^{K}$ is
\begin{equation}
  \Psi_{A^{K}}=\frac{1}{ML}\sum_{m,l,l'}\ket{\psi^{m,l}}\bra{\psi^{m,l'}}_{A^K}
\end{equation}
and therefore has cost $N\bra{\psi}G\ket{\psi}$ by \cref{eq:costQuantumPulse}. \newline~\newline%
\textit{Decoding}: After $K$ uses of the isometric extension, the overall
state is
\begin{equation}
\ket{\Psi}_{RB^{K}E^{K}}=\frac{1}{\sqrt{ML}}\sum_{m,l}\ket{m}_{R}%
\ket{\psi^{m,l}}_{B^{K}E^{K}},
\end{equation}
where $\ket{\psi^{m,l}}_{B^{K}E^{K}}\equiv U_{A\rightarrow BE}^{\ot
K}\ket{\psi^{m,l}}_{A^{K}}$. Let $\{\Lambda_{B^{K}}^{m,l}\}_{m,l}$ denote the
POVM used as a decoder in the private communication protocol. A coherent version of this
measurement is given by the isometry $V_{B^{K}\rightarrow B^{K}\hat{B}_{0}%
\hat{B}_{1}}=\sum_{m,l}\sqrt{\Lambda_{B^{K}}^{m,l}}\otimes|m\rangle_{\hat
{B}_{0}}\otimes|l\rangle_{\hat{B}_{1}}$, and after performing it, the
resulting state is approximately equal to the following one:%
\begin{equation}
\ket{\Psi}_{RB^{K}E^{K}\hat{B}_{0}\hat{B}_{1}}=\frac{1}{\sqrt{ML}}%
\sum_{m,l}\ket{m}_{R}\ket{\psi^{m,l}}_{B^{K}E^{K}}|m\rangle_{\hat{B}_{0}%
}|l\rangle_{\hat{B}_{1}}.%
\end{equation}
At this point, we know from the privacy condition for the private code, that for each $m$, the
following approximation holds: $\frac{1}{L}\sum_{l}\psi_{E^{K}}^{m,l}%
\approx\mathcal{N}^{c}(\psi^{0})^{\otimes K}$. Thus, given $m$, the eavesdropper's system is approximately independent of $m$.
By Uhlmann's theorem~\cite{U76}, for each $m$, there exists an isometry $W^m_{B^{K}\hat{B}%
_{1}\rightarrow\hat{B}_{2}}$ such that
\begin{equation}
W^m_{B^{K}\hat{B}_{1}\rightarrow\hat
{B}_{2}}\left[\frac{1}{\sqrt{L}}\sum_{l}\ket{\psi^{m,l}}_{B^{K}E^{K}}|l\rangle
_{\hat{B}_{1}}\right]\approx|\varsigma\rangle_{E^{K}\hat{B}_{2}},
\end{equation}
where
$|\varsigma\rangle_{E^{K}\hat{B}_{2}}$ is some state independent of $m$. So
this means that the receiver can perform the controlled isometry $\sum
_{m}|m\rangle\langle m|_{\hat{B}_{0}}\otimes W^m_{B^{K}\hat{B}_{1}%
\rightarrow\hat{B}_{2}}$, and the resulting state is approximately close to
the following state:%
\begin{equation}
\frac{1}{\sqrt{M}}\sum_{m}\ket{m}_{R}|m\rangle_{\hat{B}_{0}}\otimes
|\varsigma\rangle_{E^{K}\hat{B}_{2}}.%
\end{equation}
By the properties of the PPM scheme for private communication, the quantum rate per unit cost of this scheme is equal to
$ N_{\mathcal{N}}(\psi,\psi^0)/\bra{\psi} G \ket{\psi}$.

Now, if $\psi$ is not orthogonal to $\psi^0$, we implement the above protocol but replacing $\psi^{\ot N}$ with its normalized rejection from $(\psi^0)^{\ot N}$. That is, we take the component of $\ket{\psi}^{\ot N}$ orthogonal to $\ket{\psi^0}^{\ot N}$ and normalize it. Calling this $\psi^\perp_{A^N}$, we have
\begin{equation}
  \ket{\psi^\perp} = \frac{1}{\sqrt{1 - \abs{\braket{\psi^0 \vert \psi}}^{2N}}} \left(\ket{\psi}^{\ot N} - \braket{\psi^0 \vert \psi}^N \ket{\psi^0}^{\ot N}\right),
\end{equation}
and so
\begin{equation}
  \frac{1}{2} \norm{ \psi^\perp_{A^N} - \left( \psi^{\ot N} \right)_{A^N}}_1 =  \abs{\langle \psi^0 \vert \psi\rangle }^{N} \equiv \delta_N. 
\end{equation}
By monotonicity, the trace distance is at most $\delta_N$ after $N$ uses of $\mathcal{N}$ or $\mathcal{N}^c$. Hence, since $\psi$ is positive-cost and so $\abs{\braket{\psi^0 \vert \psi}} <1$, by using $\psi^{\perp}$ we expect to obtain the desired rate in the limit of large $N$. Indeed, since $\psi^\perp$ is orthogonal to $(\psi^0)^{\ot N}$, we can implement the above protocol and achieve a quantum rate per unit cost
\begin{multline}
   \Big[D_H^\varepsilon \left( \mathcal{N}^{\ot N} (\psi^\perp) \Vert \mathcal{N}(\psi^0)^{\ot N} \right) \\
   - D_{\max}^{\varepsilon'} \left( (\mathcal{N}^c)^{\ot N} (\psi^\perp) \Vert \mathcal{N}^c(\psi^0)^{\ot N} \right)\Big] / \left(\bra{\psi^\perp}G_N \ket{\psi^\perp} \right)
\end{multline}
for any $\varepsilon, \varepsilon' >0$, where $D_H^\varepsilon$ is the hypothesis testing relative entropy~\cite{wang2012one,buscemi2010quantum}:
\begin{equation}
  D_H^\varepsilon(\rho \Vert \sigma) \equiv -\log_2 \inf_{\substack{0 \le \Lambda \le I \\ \tr[\Lambda \rho] \ge 1-\varepsilon}} \tr[ \Lambda \sigma].
\end{equation}
Note that this is simply the negative logarithm of the quantity defined in \cref{eq:minType2}. For sufficiently large $N$, we have that $\delta_N < \varepsilon$. Then, by Lemma 7 in~\cite{datta2016second},
\begin{align}
  & D_H^\varepsilon \left( \mathcal{N}^{\ot N} (\psi^\perp) \Vert \mathcal{N}(\psi^0)^{\ot N} \right) \nonumber \\
  & \ge D_H^{\varepsilon - \delta_N} (\mathcal{N}(\psi)^{\ot N} \Vert \mathcal{N}(\psi^0)^{\ot N}) \\
  & \ge D_H^{\delta(\varepsilon)}(\mathcal{N}(\psi)^{\ot N} \Vert \mathcal{N}(\psi^0)^{\ot N}),
\end{align}
where $\delta(\varepsilon) = \varepsilon - \delta_{N(\varepsilon)}$ and  $N(\varepsilon)$ is the smallest integer such that $\delta_{N(\varepsilon)} < \varepsilon$. Furthermore, by the definition of smooth max-relative entropy in \cref{eq:defnSmoothMax} and the triangle inequality, for sufficiently large $N$ such that $\delta_N < \varepsilon'$, we have
\begin{align}
  & D_{\max}^{\varepsilon'} ( (\mathcal{N}^c)^{\ot N} (\psi^\perp) \Vert \mathcal{N}^c(\psi^0)^{\ot N} ) \nonumber \\
  & \le D_{\max}^{\varepsilon' - \delta_N} ( \mathcal{N}^c(\psi)^{\ot N} \Vert \mathcal{N}^c(\psi^0)^{\ot N} ) \\
  & \le D_{\max}^{\delta(\varepsilon')} ( \mathcal{N}^c(\psi)^{\ot N} \Vert \mathcal{N}^c(\psi^0)^{\ot N} ).
\end{align}
Lastly,
\begin{align}
  \bra{\psi^\perp} G_N \ket{\psi^\perp} & = \frac{1}{1- \abs{ \braket{\psi^0 \vert \psi }}^{2N}} \bra{\psi}^{\ot N} G_N \ket{\psi}^{\ot N} \\
  & = \frac{N\bra{\psi} G \ket{\psi}}{1- \abs{ \braket{\psi^0 \vert \psi }}^{2N}}.
\end{align}
We conclude that we can achieve
\begin{multline}
  \Big( 1  -  \abs{ \braket{\psi^0 \vert \psi }}^{2N} \Big)
  \Big[  D_H^{\delta(\varepsilon)}(\mathcal{N}(\psi)^{\ot N} \Vert \mathcal{N}(\psi^0)^{\ot N}) \\
   - D_{\max}^{\delta(\varepsilon')}(\mathcal{N}^c(\psi)^{\ot N} \Vert \mathcal{N}^c(\psi^0)^{\ot N})\Big] / \left(N\bra{\psi} G \ket{\psi}\right).
\end{multline}
By Quantum Stein's Lemma (\cref{eq:quantumStein}) and the quantum asymptotic equipartition property for smooth max-relative entropy, in the limit of large $N$ we can therefore achieve a quantum rate per unit cost of $\frac{N_\mathcal{N}(\psi, \psi^0)}{\bra{\psi} G \ket{\psi}}$.

Note that just like in the private case, this scheme gives achievability for non-degradable channels as well.

\section{Capacities per Unit Cost of Quantum Gaussian Channels}

\label{sec:gaussian}

By the methods developed and used in \cite{holevo2003entanglement,H04,holevo2013quantum,wilde2016energy}, we can generalize the above results to infinite dimensions, in particular for quantum Gaussian channels. We will use the formulas derived above to compute various capacities per unit cost for specific Gaussian channels, where the cost observable is the photon number operator $\hat{n} = \sum_{n=0}^{\infty} n \state{n}$, with $\ket{n}$ being a photon number state. Since the channels we consider already have known capacity cost functions, it is easiest for us to compute the capacity per unit cost via the following formula:
\begin{equation}
 \textbf{\textit{C}}(\mathcal{N})= \lim_{\bar n \to 0} \frac{C(\mathcal{N},\bar n)}{\bar n}. 
 \label{eq:diff-cap-cost}
\end{equation}
Note that we could also compute the capacities per unit cost using the optimized relative entropy formulas that we obtained in the previous sections.

The quantum Gaussian channels we consider are the following~\cite{holevo2002sending,holevo2001evaluating,giovannetti2014ultimate}. The thermal channel $\mathcal{E}_\eta^{n_\text{th}}$ with transmissivity $\eta\in(0,1)$ and thermal photon number $n_\text{th}\in\mathbb{R}_{>0}$ is a Gaussian channel which mixes the input signal with a thermal Gaussian state. This is summarized by the following Heisenberg evolution:
\begin{equation}
  a_\text{in} \mapsto \sqrt{\eta} a_\text{in} + \sqrt{1-\eta} a_\text{th},
\end{equation}
where $a_\text{in}, \, a_\text{th}$ are the annihilation operators for the input and thermal modes, respectively. We also consider the additive classical noise channel $\mathcal{N}_N$ with variance $N\in\mathbb{R}_{>0}$, which describes classical noise that displaces the signal in phase space according to a Gaussian distribution. The Heisenberg evolution is
\begin{equation}
  a_\text{in} \mapsto a_\text{in} + \xi,
\end{equation}
where $\xi$ is a complex normal random variable with mean zero and variance $N$. Next, the amplifier channel $\mathcal{A}_\kappa^{n_\text{th}}$ with gain parameter $\kappa > 1$ and thermal photon number $n_\text{th}\in\mathbb{R}_{>0}$ describes the effect of a two-mode squeezing Hamiltonian that acts on the input mode and a thermal mode. This effectively amplifies the input signal but at the cost of adding noise. The resulting Heisenberg evolution is given by
\begin{equation}
  a_\text{in} \mapsto \sqrt{\kappa} a_\text{in} + \sqrt{\kappa-1} a_\text{th}^\dagger.
\end{equation}
Lastly, we also look at the weak conjugate of the amplifier channel, the contravariant amplifier $\tilde{\mathcal{A}}^{n_\text{th}}_\kappa$ with gain parameter $\kappa > 1$, and thermal photon number $n_\text{th} \in \mathbb{R}_{>0}$. 

\subsection{Classical Communication over Gaussian Channels}
We compute the unassisted classical capacities per unit cost for these four channels. For the thermal channel, the classical capacity cost function is given by \cite{GHG15,giovannetti2014ultimate}
\begin{equation}
  C(\mathcal{E}_\eta^{n_\text{th}}, \bar n) = g(\eta \bar n + (1-\eta) n_\text{th}) - g((1-\eta) n_\text{th}),
\end{equation}
where 
\begin{equation}
  g(x) \equiv (x+1) \log_2(x+1) - x \log_2 x.
\end{equation}
Hence, by applying \cref{eq:diff-cap-cost}, we find the following:
\begin{align}
    \textbf{\textit{C}}(\mathcal{E}_\eta^{n_\text{th}}) &= \eta \log_2 \!\left( 1+ \frac{1}{n_\text{th}(1-\eta)}\right) \label{eq:capacity_per_unit_cost_thermal_channel}.
\end{align}

Now, consider the additive-noise channel, which has the following classical capacity cost function \cite{GHG15,giovannetti2014ultimate}
\begin{equation}
  C(\mathcal{N}_N, \bar n) = g(\bar n + N)- g(N).
\end{equation}
Hence,
\begin{equation}
  \textbf{\textit{C}}(\mathcal{N}_N) = \log_2\!\left(  1+ \frac{1}{N}\right).
\end{equation}
Next, we consider the amplifier channel, which has a classical capacity cost function \cite{GHG15,giovannetti2014ultimate}
\begin{equation}
  C(\mathcal{A}_\kappa^{n_\text{th}},\bar n) = g(\kappa \bar n + (\kappa-1)(n_\text{th}+1)) - g( (\kappa-1) (n_\text{th}+1)).
\end{equation}
We find
\begin{equation}
  \textbf{\textit{C}}(\mathcal{A}_\kappa^{n_\text{th}}) = \kappa \log_2 \!\left( 1+ \frac{1}{(\kappa-1)(n_\text{th}+1)} \right).
\end{equation}
Finally, for the contravariant amplifier channel, which has the capacity cost function~\cite{giovannetti2014ultimate}
\begin{equation}
  C(\tilde{\mathcal{A}}^{n_\text{th}}_\kappa, \bar n) = g(\kappa n_\text{th} + (\kappa-1) (\bar n+1)) - g(\kappa (n_\text{th}+1)-1),
\end{equation}
we have
\begin{equation}
  \textbf{\textit{C}}(\tilde{\mathcal{A}}^{n_\text{th}}_\kappa) = (\kappa-1) \log_2\!\left( 1+\frac{1}{\kappa(n_\text{th}+1)-1} \right).
\end{equation}

Note that given \cref{eq:monoZero} and the achievability result in \cref{thm:generalCCost}, when a zero-cost state exists, we can achieve the capacity per unit cost with any code that achieves the cost-constrained capacity in the limit of zero cost. For example, for the pure-loss bosonic channel, single-photon-detection achieves the classical capacity in the limit $\bar n \to 0$~\cite{takeoka2014capacity}. In general, we can achieve the capacity per unit cost with any code that achieves $C(\beta_{\max})$ where $\beta_{\max} = \arg \sup_{\beta > 0} C(\beta)/\beta$. 

\subsection{Entanglement-Assisted Communication over Gaussian Channels}
We now compute the entanglement-assisted capacity per unit cost for the first three channels. 
For the thermal channel, the entanglement-assisted capacity cost function is given by~\cite{holevo2001evaluating,holevo2012quantum}
\begin{align}
  & C_\mathrm{EA}(\mathcal{E}^{n_\text{th}}_\eta, \bar n)  \nonumber \\
  & = g(\bar n) + g(\eta \bar n  + (1-\eta)n_\text{th}) \nonumber\\
  & -g\Big( \frac{1}{2}\Big( \sqrt{((1+\eta) \bar n + (1-\eta)n_\text{th} + 1)^2 - 4 \eta \bar n(\bar n+1)} \nonumber\\
  & \qquad \qquad - (1-\eta)(\bar n - n_\text{th}) -1 \Big) \Big) \nonumber\\
  & - g\Big( \frac{1}{2}\Big( \sqrt{((1+\eta) \bar n + (1-\eta)n_\text{th} + 1)^2  - 4 \eta \bar n(\bar n+1)} \nonumber\\
  & \qquad\qquad + (1-\eta)(\bar n - n_\text{th}) -1 \Big) \Big) .
\end{align}
We compute the limit as per \cref{eq:diff-cap-cost} and find that the entanglement-assisted capacity per unit cost of the thermal channel diverges\footnote{For many of these calculations, see the Mathematica file included in arXiv posting.}. For the additive noise channel, the entanglement-assisted capacity cost function is~\cite{holevo2001evaluating,holevo2012quantum}
\begin{multline}
   C_\mathrm{EA}(\mathcal{N}_N, \bar n)  \\
   = g(\bar n) + g(\bar n + N) -g\left( \frac{1}{2} \left( \sqrt{ (N +1)^2 + 4N \bar n} - N  -1 \right) \right)  \\
     - g\left( \frac{1}{2} \left( \sqrt{(N+1)^2 + 4N \bar n} + N -1 \right) \right) ,
\end{multline}
which also gives an infinite capacity per unit cost. Lastly, for the amplifier channel,
\begin{align}
   & C_\mathrm{EA}(\mathcal{A}_\kappa^{n_\text{th}},\bar n) \nonumber \\
   & = g(\bar n) + g( \kappa \bar n + (\kappa -1) (n_\text{th}+1))  \nonumber\\
   & - g\Big( \frac{1}{2} \Big( \sqrt{( (\kappa+1) \bar n + (\kappa -1) (n_\text{th}+1) + 1)^2 - 4 \kappa \bar n (\bar n+1)} \nonumber \\
   & \qquad\qquad - (\kappa-1) (\bar n + n_\text{th}+1) -1 \Big) \Big)  \nonumber \\
   &- g\Big( \frac{1}{2} \Big( \sqrt{( (\kappa+1) \bar n + (\kappa -1) (n_\text{th}+1) + 1)^2 - 4 \kappa \bar n (\bar n+1)} \nonumber \\
   & \qquad\qquad + (\kappa-1) (\bar n + n_\text{th}+1) -1 \Big) \Big).
\end{align}
Computing the ratio again gives an infinite value as $\bar n \to 0$. 

We show the divergence of $\textbf{\textit{C}}_\mathrm{EA}$ by plotting the bits per photon against the number of photons for these three Gaussian channels in \cref{fig:EAdivergence}.
\begin{figure}
  \centering
  \includegraphics[width= 0.45\textwidth]{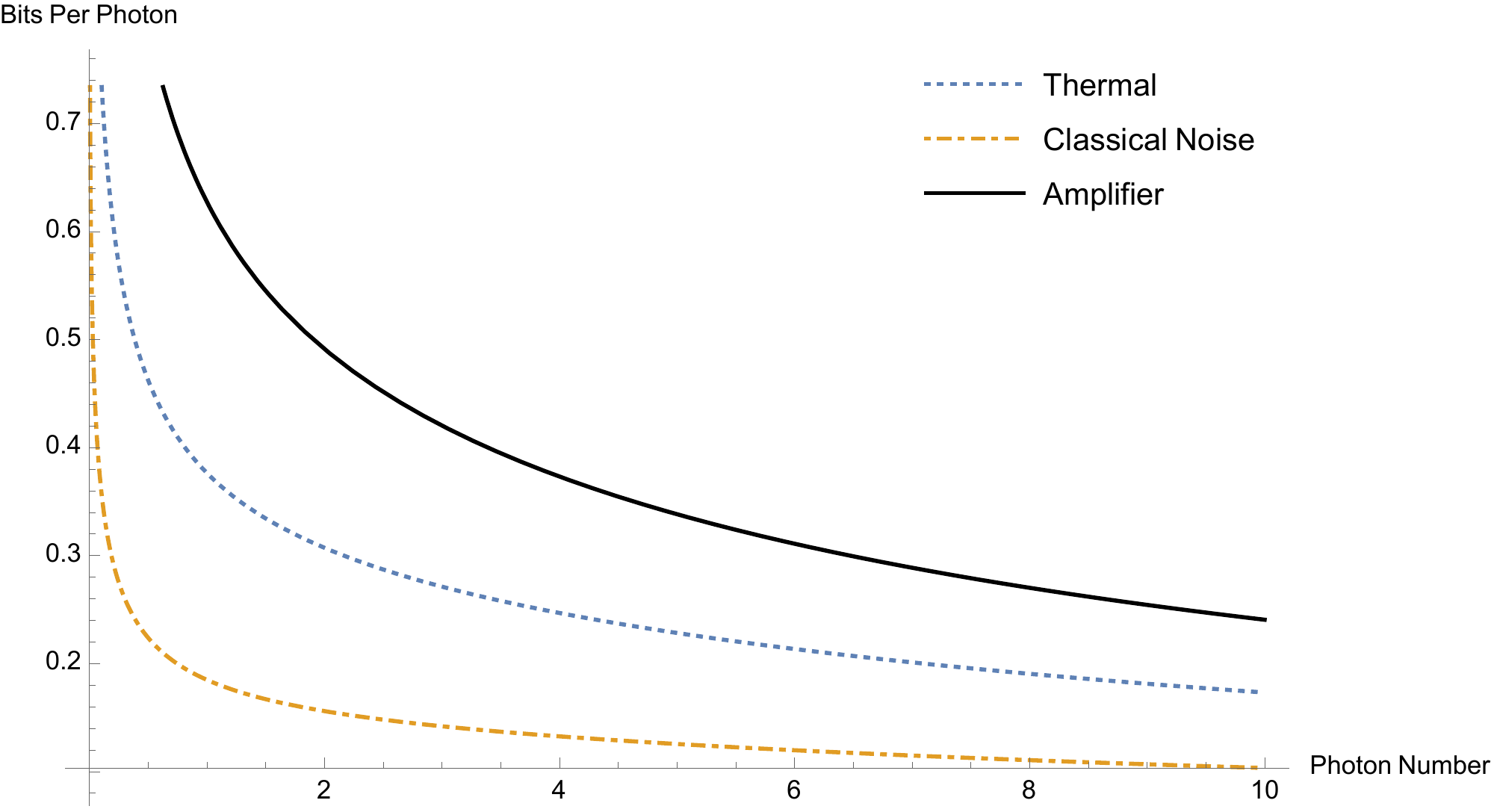}
  \caption{The figure illustrates the divergence of $C_\mathrm{EA}(\mathcal{N},\bar n)/\bar n$ as $\bar n \to 0$. For the thermal channel we set $n_\text{th} = 10,\, \eta = 0.7$, for the classical noise channel $N = 10$, and for the amplifier channel $n_\text{th}=10, \, \kappa = 1.3$.}
  \label{fig:EAdivergence}
\end{figure}
It is likely that the divergences for these Gaussian channels appear since we are allowing for an unbounded amount of entanglement assistance per unit cost. However, it is not true that this potentially infinite entanglement assistance always leads to divergences. Trivial examples include channels where $C_\mathrm{EA}(\mathcal{N}, \beta) = 0$ for all $\beta$ and when the cost observable is positive definite. We can also find examples using the fact that entanglement assistance sometimes does not help. Take for instance the state preparation qubit channel
\begin{equation}
  \mathcal{N}_s : \rho \mapsto \bra{0} \rho \ket{0} \rho^0 + \bra{1} \rho \ket{1} \rho^1
\end{equation}
with cost observable $G = \state{1}$, where $\rho^0, \rho^1 \in \mathcal{S}(\mathcal{H}_{A'})$ such that $0<D(\rho^1 \Vert \rho^0) <\infty$. By Proposition 4 in~\cite{shirokov2012conditions}, $C_\mathrm{EA}(\mathcal{N}_s , \beta) = C(\mathcal{N}_s, \beta)$ for all $\beta$. Now, we have a zero-cost state $\state{0}$, so we can use \cref{thm:zero-cost-rel-ent-CC}:
\begin{align}
  \textbf{\textit{C}}(\mathcal{N}_s) & = \sup_{\psi \neq \state{0}} \frac{D(\mathcal{N}_s(\psi) \Vert \mathcal{N}_s(\state{0}))}{\bra{\psi} G \ket{\psi}}\\
  & = \sup_{\psi \neq \state{0}} \frac{D\left( \abs{c_0}^2 \rho^0 + \abs{c_1}^2 \rho^1 \Vert \rho^0  \right)}{\abs{c_1}^2}\\
  & \le \sup_{\psi \neq \state{0}} \frac{1}{\abs{c_1}^2} \left( \abs{c_0}^2 D(\rho^0 \Vert \rho^0) + \abs{c_1}^2 D(\rho^1 \Vert \rho^0)\right) \\
  & =  D(\rho^1 \Vert \rho^0) \\
  & < \infty,
\end{align}
where $\ket{\psi} \equiv c_0 \ket{0} + c_1 \ket{1}$ and we can divide by $\abs{c_1}^2$ since $\psi \neq \state{0}$. The inequality follows from the convexity of the relative entropy in the first argument. We conclude $\textbf{\textit{C}}_\mathrm{EA}(\mathcal{N}_s) = \textbf{\textit{C}}(\mathcal{N}_s) < \infty$. However, since $D(\rho^1 \Vert \rho^0) >0$, this channel clearly has a non-zero $C(\mathcal{N}_s, \beta)$ and so non-zero $C_\mathrm{EA}(\mathcal{N}_s, \beta)$. It is an interesting open question to find an explicit nontrivial channel for which there is a gain from entanglement assistance but $\textbf{\textit{C}}_\mathrm{EA}$ is finite.

\subsection{Private and Quantum Communication over Gaussian Channels}

\begin{figure}
  \centering
  \includegraphics[width =0.45\textwidth]{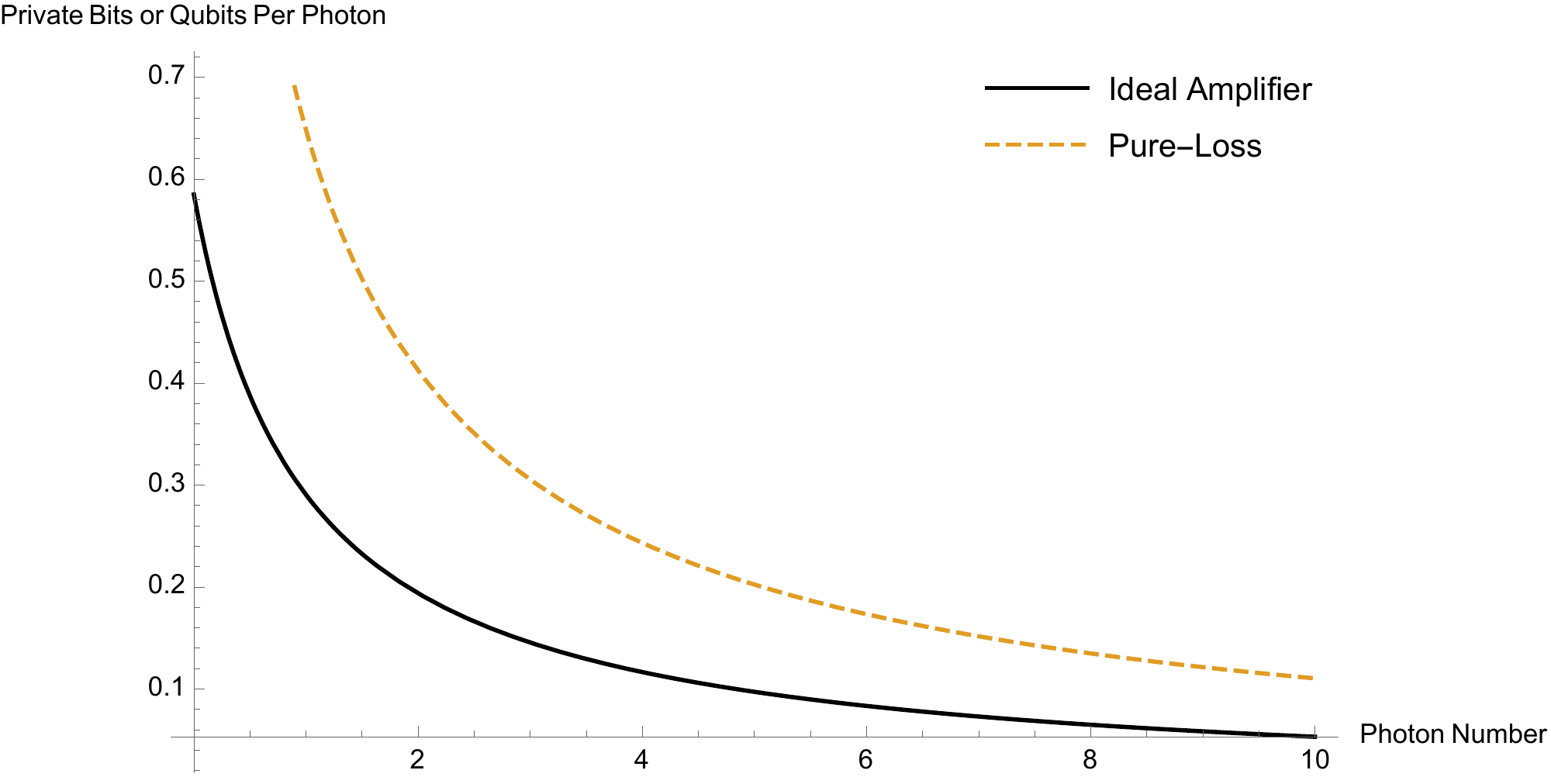}
  \caption{Plot of $P(\mathcal{N},\bar n)/\bar n = Q(\mathcal{N},\bar n)/\bar n$ versus $\bar n$. For the ideal amplifier we set $\kappa = 3$ and for the pure-loss channel $\eta = 0.7$.}
  \label{fig:gaussianPQ}
\end{figure}

We next compute the private and quantum capacities per unit cost for degradable Gaussian channels. In particular, we consider ideal amplifiers $\mathcal{A}^0_\kappa$ and pure-loss channels $\mathcal{E}^0_\eta$ for $\eta \geq 1/2$. For the former, we have the capacity cost function \cite{wilde2016energy}
\begin{equation}
  P(\mathcal{A}^0_\kappa,\bar n) = Q(\mathcal{A}^0_\kappa,\bar n) = g(\kappa (\bar n+1) -1) - g( (\kappa-1) (\bar n+1)).
\end{equation}
Hence,
\begin{equation}
  \textbf{\textit{P}}(\mathcal{A}^0_\kappa) = \textbf{\textit{Q}}(\mathcal{A}^0_\kappa) = \log_2 \!\left( \frac{\kappa}{\kappa-1} \right),
  \label{eq:cap-cost-amp-q-limited}
\end{equation}
which demonstrates that the quantum and private capacity per unit cost are equal to the unconstrained quantum and private capacity of the ideal amplifier channel \cite{WPG07,wilde2016energy}.
Finally, for the pure-loss channel, we have that \cite{wilde2016energy}
\begin{equation}
  P(\mathcal{E}^0_\eta, \bar n) = Q(\mathcal{E}^0_\eta, \bar n) = g(\eta \bar n) - g((1-\eta) \bar n).
\end{equation}
We divide by $\bar n$ and find the following divergent term for $\eta > 1/2$ as $\bar n \to 0$:
\begin{equation}
  (2\eta-1) \log_2 (1/\bar n).
\end{equation}
Hence, the pure-loss channel for $\eta >1/2$ has an infinite private and quantum capacity per unit cost, while for $\eta \le 1/2$ the capacities per unit cost are zero since the channel is antidegradable in this regime. In \cref{fig:gaussianPQ} we plot the private bits or qubits per photon against the number of photons for these two channels.

Another interesting communication setting that we can consider is that of private and quantum communication assisted by a side classical communication channel, their corresponding cost-constrained capacities denoted by $P_2(\mathcal{N}, \beta)$ and $  Q_2(\mathcal{N},\beta)$, for which a general theory has  been developed recently in \cite{DSW18}. Just as for the other communication settings considered in this paper, we could trivially take an achievability result for the cost-constrained capacity to obtain one for the capacity per unit cost. For something less trivial, we can consider the  example of the pure-loss bosonic channel, for which it is known~\cite{takeoka2015fundamental} that $\lim_{\bar n \to \infty} P_2(\mathcal{E}^0_\eta, \bar n), Q_2(\mathcal{E}^0_\eta, \bar n)$ is finite. This is a pessimistic result which can be interpreted to imply that the rate of quantum key distribution over a fiber-optic cable is finite even if one uses arbitrarily high input energy. However, the situation is different for capacity per unit cost. 
Previously we calculated that the unassisted private and quantum capacities per unit cost of this channel are infinite, and since these trivially lower bound the two-way assisted capacities, we conclude that $\textbf{\textit{P}}_2(\mathcal{E}_\eta^0), \textbf{\textit{Q}}_2(\mathcal{E}_\eta^0)$ are infinite.

The situation for these capacities per unit cost for the quantum-limited amplifier channel $\mathcal{A}^0_\kappa$ of gain $\kappa>1$ is more interesting. 
In this case, we can use the following upper bound from \cite{goodEW16} (see also \cite[Remark~3]{DSW18})
\begin{multline}
   P_2(\mathcal{A}^0_\kappa, \bar n) 
   \le E_\text{sq}(\mathcal{A}^0_\kappa, \bar n) \\
   \le  g( (1+\kappa) \bar n / 2 + (\kappa-1)/2) - 
   g( (\kappa - 1) (\bar n +1) / 2).
\end{multline}
where $E_\text{sq}(\mathcal{A}^0_\kappa, \bar n)$ denotes the energy-constrained squashed entanglement of the channel~\cite{DSW18}. This was used to show that the private capacity in the limit of infinite photon number is finite. For this channel this is true for the private capacity per unit cost as well:
\begin{align}
  & \textbf{\textit{P}}_2(\mathcal{A}^0_\kappa) \nonumber \\
  & = \lim_{\bar n \to 0} \frac{P_2(\mathcal{A}^0_\kappa, \bar n)}{\bar n} \\
  & \le \lim_{\bar n \to 0} \frac{ g( (1+\kappa) \bar n / 2 + (\kappa-1)/2) - 
   g( (\kappa - 1) (\bar n +1) / 2)}{\bar n} \\
  & = \log_2\!\left(\frac{\kappa+1}{\kappa-1}\right).
\end{align}
Note that without any improved upper bounds on $P_2(\mathcal{A}^0_\kappa, \bar n)$, there thus remains a gap between the lower bound from \eqref{eq:cap-cost-amp-q-limited} and the upper bound given above. Also, since private capacity bounds quantum capacity from above, the same conclusions follow for the quantum capacity.

\section{Blocklength-Constrained Capacity Per Unit Cost}
\label{sec:blockConstraint}

In \cref{sec:gaussian}, we find many infinite capacities per unit cost. Although we are able to ascribe these infinities to unphysical assumptions such as unbounded entanglement assistance, we can alternatively ascribe almost any infinite capacity per unit cost to the \textit{absence of a time constraint}. The key is to notice that although \cref{dfn:capacityPerCost} is very similar to \cref{dfn:constrainedCapacity}, there is a significant asymmetry: while the latter has a constraint on the cost divided by blocklength, the former does not have an analogous constraint. Indeed, infinite capacities is not a new phenomenon: it was encountered in the study of continuous variable channels, which was part of the motivation of studying cost constraints in the first place. Hence, \emph{just as cost constraints tamed infinite capacities for continuous variable channels, we likewise expect that blocklength constraints would tame infinite capacities per unit cost}.

Indeed, we can show that the former implies the latter. We first define the blocklength-constrained capacity per unit cost. For concreteness, we  state everything in the unassisted classical communication setting.
\begin{dfn}
  \label{dfn:constrainedCapacityPerCost}
  A non-negative number $\textbf{\textit{R}}$ is an achievable rate per unit cost with blocklength constraint $\alpha \in \R_{> 0}$ if for every $\varepsilon \in (0,1), \delta > 0$, $\exists \, \nu_0 > 0$ such that if  $\nu \ge \nu_0$, $\exists$ an $(n, M, \nu, \varepsilon)$ code such that $\log_2 M > \nu (\textbf{\textit{R}}- \delta)$ and $n \le \nu \alpha$. The capacity per unit cost with blocklength constraint $\alpha$ is the supremum of all achievable rates per unit cost with constraint $\alpha$, denoted as $\textbf{\textit{C}}(\mathcal{N},\alpha)$.
\end{dfn}
Conveniently, we can show that $\textbf{\textit{C}}(\mathcal{N},\alpha)$ has a characterization similar to that of $\textbf{\textit{C}}(\mathcal{N})$. The proof of this theorem uses results from the proof of \cref{thm:generalCCost} and that of the corollary in~\cite{verdu1990channel}.
\begin{thm}
  \label{thm:constrainedCapacityPerCost}
  For blocklength constraint $\alpha > 0$,
  the blocklength-constrained capacity per unit cost of a quantum channel $\mathcal{N}$ is given by
  \begin{equation}
    \textbf{\textit{C}}(\mathcal{N}, \alpha) = \sup_{\beta \ge \frac{1}{\alpha}} \frac{C(\mathcal{N},\beta)}{\beta}. 
  \end{equation}
\end{thm}
\begin{proof}
  We first show achievability. Let $\beta >0$. By definition, we can achieve a rate $C(\mathcal{N},\beta)$ with cost $\beta$. 
  Let $\varepsilon \in (0,1), \delta > 0$. 
  Then, using the direct part in the proof of \cref{thm:generalCCost}, 
  $\exists \, \nu_0$ such that for $\nu \ge \nu_0$ we can find an $(n,M,\nu,\varepsilon)$ code where
  \begin{equation}
    \frac{\log_2 M}{\nu} > \frac{C(\mathcal{N},\beta)}{\beta} - \delta  
  \end{equation}
  and
  \begin{equation}
    n \le \frac{\nu}{\beta}.
  \end{equation}
  Hence, for $\beta \ge 1/\alpha$, we achieve a rate per unit cost $C(\mathcal{N},\beta)/\beta$ with blocklength constraint $\frac{1}{\beta} \le \alpha$.

  We now show the converse. We claim that if we can achieve a rate per unit cost $\textbf{\textit{R}}$ with blocklength constraint $\alpha$, then we can achieve a rate $\textbf{\textit{R}}/\alpha$ at cost $1/\alpha$. 

  To see this, let $\varepsilon \in (0,1), \delta >0$. Then, by \cref{dfn:constrainedCapacityPerCost}, there exists $\nu_0$ such that for all $\nu \ge \nu_0$, there is an $(n,M,\nu, \varepsilon)$ code where $n \le \nu \alpha$ and
\begin{equation}
  \frac{\log_2 M}{\nu} > \textbf{\textit{R}} - \alpha \delta.
\end{equation}
Thus, this is an $(n,M, \nu,\varepsilon)$ code that satisfies
\begin{equation}
  \frac{\log_2 M}{n} \ge \frac{\log_2 M}{\nu \alpha} > \frac{\textbf{\textit{R}}}{\alpha} - \delta.
\end{equation}
Let $n_0 \equiv \nu_0 \alpha$ and $n \ge n_0$. Then, let $\nu = n /\alpha$. This implies $\nu \ge n_0 /\alpha = \nu_0$, so the above code with this $\nu$ achieves a rate per channel use $\textbf{\textit{R}}/\alpha$ and also has an average cost $1/\alpha$. This establishes the claim.

We now proceed by contradiction. Suppose we can achieve a rate per unit cost $\textbf{\textit{R}} > \sup_{\beta \ge \frac{1}{\alpha}} \frac{C(\mathcal{N},\beta)}{\beta}$ with constraint $\alpha$. Then, by the claim we can achieve the following rate with cost $1/\alpha$:
\begin{align}
  \frac{\textbf{\textit{R}}}{\alpha} > \sup_{\beta \ge \frac{1}{\alpha}} \frac{C(\mathcal{N},\beta)}{\beta \alpha} \ge C(\mathcal{N},1/\alpha).
\end{align}
Since this contradicts the definition of the capacity cost function, the converse follows.
\end{proof}
In particular, when there is a zero-cost state, by concavity $C(\mathcal{N},\beta)/\beta$ is monotone non-increasing on $(0, + \infty)$ and so
\begin{equation}
  \label{eq:blocklengthConstrainedZeroCost}
  \textbf{\textit{C}}(\mathcal{N},\alpha) = \alpha C(\mathcal{N},1/\alpha).
\end{equation}
This therefore removes the infinities, assuming that the capacity cost function does not diverge at finite cost. Also, note that \cref{eq:blocklengthConstrainedZeroCost} gives an operational interpretation of the ratio of the cost-constrained capacity to the cost: When there is a zero-cost state, for any $\beta >0$, $C(\mathcal{N}, \beta) /\beta$ is the capacity per unit cost with blocklength constraint $1/\beta$. Furthermore, by inverting \cref{eq:blocklengthConstrainedZeroCost}, we obtain an operational interpretation of the ratio of the blocklength-constrained capacity per unit cost to the blocklength constraint. 

Interestingly, even in the general case where there might not be a zero-cost state, we can express the cost-constrained capacity in terms of the blocklength-constrained capacity per unit cost. Indeed, by the same proof of \cref{thm:constrainedCapacityPerCost} except simply switching the achievability and the converse, we can show that
\begin{equation}
  C(\mathcal{N},\beta) = \sup_{\alpha \ge \frac{1}{\beta}} \frac{\textbf{\textit{C}}(\mathcal{N},\alpha)}{\alpha}.
\end{equation}
This expression and \cref{thm:constrainedCapacityPerCost} therefore establish a duality between the two quantities, which we can refer to as the \textit{blocklength-cost duality}. We can find many instances of this duality by exchanging ``blocklength'' with ``cost'' and vice versa.

\subsection{The Infinite Capacity per Unit Photon of the Pure-loss Bosonic Channel}
In the limit $n_\text{th} \to 0$, the thermal channel reduces to the pure-loss bosonic channel $\mathcal{L}_\eta$~\cite{giovannetti2004classical}, which has classical capacity
\begin{equation}
  C(\mathcal{L}_\eta, \bar n) = g(\eta \bar n).
\end{equation}
With a simple calculation we find that $\textbf{\textit{C}}(\mathcal{L}_\eta)$ is infinite, an observation made in~\cite{guha2011quantum} and earlier in the context of infinite bandwidth~\cite{lebedev1966information,caves1994quantum,yuen1975optimum}\footnote{For a connection between infinite bandwidth capacity and capacity per unit cost, see~\cite{verdu1990channel}.}. This is a uniquely quantum phenomenon as there is no direct analogue of the pure-loss channel in classical information theory, and this therefore raised the debate of whether quantum mechanics somehow unlocks the ability to achieve arbitrarily high rates of communication. Our answer to this is negative, and as mentioned above, we propose a solution to this and other infinite capacities per unit cost by introducing a blocklength constraint. In particular, just as not having a photon number constraint led to diverging capacities of bosonic channels, not having a blocklength constraint can lead to diverging capacities per unit cost. The unphysical assumption in the former is that infinite input power is available, while that of the latter is that infinite \emph{time} is available. 

This observation motivates considering a composite cost observable given by the sum of time and energy:
\begin{align}
  G \equiv I + \hat n.
\end{align}
This way, we can remove both unphysical assumptions at once by effectively constraining both time \textit{and} energy. And we indeed find that with respect to this cost observable, the capacity per unit cost is finite. To see this, let $\mathcal{L}_\eta$ be for instance the pure-loss bosonic channel. Then,
\begin{align}
  \textbf{\textit{C}} (\mathcal{L}_\eta) & = \sup_{\beta > 0} \frac{C(\mathcal{L}_\eta, \beta)}{\beta} \nonumber \\
  & = \sup_{\beta > 1} \frac{C_{G= \hat n}(\mathcal{L}_\eta, \beta-1)}{\beta} \nonumber \\
  & = \sup_{\beta > 1} \frac{g(\eta (\beta-1))}{\beta}, \label{eq:lossLimited}
\end{align}
where we take the supremum over $\beta >1$ since otherwise the numerator is zero, and 
\begin{equation}
C_{G = I+ \hat n}(\mathcal{L}_\eta, \beta) = C_{G = \hat n}(\mathcal{L}_\eta, \beta-1)
\end{equation}
since any protocol with average cost $\beta$ with respect to cost $I + \hat n$ is also a protocol with average cost $\beta-1$ with respect to $\hat n$. Now, 
\begin{align}
  \lim_{\beta \to 1} \frac{g(\eta (\beta-1))}{\beta} & = 0,\\
  \lim_{\beta \to \infty} \frac{g(\eta (\beta-1))}{\beta} & = \lim_{\beta \to \infty} \frac{\eta \log\left(1+ \frac{1}{\eta(\beta-1)}\right)}{1} = 0.
\end{align}
Furthermore,~\cref{eq:lossLimited} is continuous for $\beta \in (1, \infty)$. Hence, it is bounded and thus $\textbf{\textit{C}}(\mathcal{L}_\eta)$ with respect to the cost observable $I + \hat n$ is finite as claimed. Thus, if we have constraints on both time and energy, the capacity is finite. This is superficially reminiscent of the time-energy uncertainty principle in quantum mechanics: just as the individual precisions of time and energy can be unbounded but their product cannot, the capacities with respect to time and energy can be unbounded but the capacity with respect to their sum cannot.

Another way to make sense of the infinite capacity per unit cost of the pure-loss channel is to look at $\textbf{\textit{C}}(\mathcal{E}^{n_\text{th}}_\eta)$ for small $n_\text{th}$:
\begin{align}
  \textbf{\textit{C}}(\mathcal{E}_\eta^{n_\text{th}})&= -\eta \log_2(n_\text{th}(1-\eta)) + \mathcal{O}(n_\text{th}). \label{eq:capacity_per_unit_cost_thermal_channel_small_noise_series}
\end{align}
This is finite for any $n_\text{th} > 0$, but diverges as $-\log_2(n_\text{th})$ as $n_\text{th} \to 0$. Since $n_\text{th}$ physically corresponds to temperature, this suggests that the infinity comes about since zero temperature is unphysical. It also implies that the capacity per unit cost can be increased arbitrarily by going to lower temperatures. We find that a toy model of a classical binary channel reproduces this qualitative difference between the zero and non-zero temperature cases and the logarithmic scaling of the capacity per unit cost in the noise parameter of the channel:
\begin{align}
\begin{tikzpicture}[baseline=(current bounding box.center)]
    \node (A) at (0,0) {$0$};
    \node (B) at (0,3) {$1$};
    \node (C) at (3,0) {$0$};
    \node (D) at (3,3) {$1$};
    \path[->]
    (A) edge node[below] {$1-\delta$} (C)
    (B) edge node[above,sloped, near start] {$\epsilon$} (C)
    (A) edge node[above,sloped, near start] {$\delta$} (D)
    (B) edge node[above] {$1-\epsilon$} (D);
\end{tikzpicture}
    \quad
\approx
\quad
\begin{tikzpicture}[baseline=(current bounding box.center)]
    \node (A) at (0,0) {$|0\rangle \langle 0|$};
    \node (B) at (0,3) {$|\alpha\rangle \langle \alpha|$};
    \node (C) at (3,0) {$0$};
    \node (D) at (3,3) {$1$};
    \path[->]
    (A) edge node[below] {$1-\delta(n_\text{th})$} (C)
    (B) edge node[above,sloped, near start] {$e^{-|\alpha|^2}$} (C)
    (A) edge node[above,sloped, near start] {$\delta(n_\text{th})$} (D)
    (B) edge node[above] {$1-e^{-|\alpha|^2}$} (D);
\end{tikzpicture}
.
\label{eq:toy_model}
\end{align}
On the left in (\ref{eq:toy_model}) is the binary channel $\mathcal{B}_{\epsilon,\delta}$ with crossover probabilities $\epsilon$ and $\delta$, and with input `0' having cost $0$, input `1' having cost $1$. On the right in (\ref{eq:toy_model}) is a binary channel induced by a PPM scheme for the thermal channel: the sender sends either the vacuum state $|0 \rangle \langle 0|$ at cost $0$ or the coherent state $|\alpha \rangle \langle \alpha|$ at cost $|\alpha|^2$, and the detector is a photon counter that measures either a click (`1') or no click (`0'), with crossover probabilities shown on the diagram.
A short computation shows that the capacity per unit cost for the binary channel $\mathcal{B}_{\epsilon,\delta}$ is:
\begin{equation}
    \textbf{\textit{C}}(\mathcal{B}_{\epsilon,\delta}) = -(1-\epsilon)\log_2(\delta)-\epsilon\log_2(1-\delta)-h(\epsilon), \label{eq:capacity_per_unit_cost_binary_channel}
\end{equation}
where $h(x) \equiv -x \log_2(x) - (1-x) \log_2(1-x)$ is the binary entropy function. Thus the capacity per unit cost of the classical binary channel $\mathcal{B}_{\epsilon,\delta}$ diverges as $-\log_2(\delta)$ as $\delta \rightarrow 0$, analogous to the behavior of (\ref{eq:capacity_per_unit_cost_thermal_channel_small_noise_series}) as $n_\text{th} \rightarrow 0$.

\section{Discussion}

In this paper we generalized the notion of capacity per unit cost to quantum channels for the tasks of classical communication, entanglement-assisted classical communication, private communication, and quantum communication.
There are some simple extensions of our results that hold but for which we do not provide details. 
First, for private capacity per unit cost, we could also consider, as in \cite{1050633}, the more general case of a quantum wiretap channel $\mathcal{N}_{A \to BE}$ where $\mathcal{N}$ is not necessarily an isometry. The results we found in Section~\ref{sec:priv-comm} apply directly to the more general case of a degraded quantum wiretap channel.

Another is to consider non-additive quantum channels, for which we will find regularized expressions of the formulas given in \cref{thm:generalCCost}, \cref{thm:zero-cost-rel-ent-CC}, \cref{thm:generalPCost}, and \cref{thm:generalQCost}. For example, for the regularized classical capacity per unit cost we obtain
\begin{equation}
  \textbf{\textit{C}}(\mathcal{N}) = \lim_{n \to \infty} \sup_{\beta >0} \frac{C(\mathcal{N}^{\ot n}, n \beta)}{n \beta}.
\end{equation}
This expression motivates an open question we raise: can additivity depend on cost constraint? For example, can the cost-constrained Holevo information $\chi(\mathcal{N}, \beta)$ be additive for some values of $\beta$ but not others? This is trivially true if we pick some non-additive channel and compare $\beta= 0$ with $\beta = \infty$, but it would be fascinating to find examples where there is some non-trivial dependence.

We think there are a number of interesting directions for future work. One direction  is to consider non-degradable quantum channels with a zero-cost state. As discussed earlier, a lower bound for the private capacity per unit cost in this case is given by our private PPM scheme and for classical channels in~\cite{el2013secrecy}, but the converse remains open. It would also be interesting to give an expression for the quantum capacity per unit cost in terms of an optimized relative entropy. Another possible direction is to prove a strong converse for the capacity per unit cost. However, as was shown in~\cite{bardhan2014strong}, the strong converse holds for Gaussian channels with an approximate peak cost constraint but not an average cost constraint. Hence, for Gaussian channels we cannot directly prove a strong converse for the classical capacity per unit cost in the same way as in \cref{thm:generalCCost}. Lastly, we can consider the regime of finite blocklength or finite-sized measurement blocks and see how much information can be sent per cost in this context.

In general, we can find a myriad of new directions by extensive use of the blocklength-cost duality. For instance, instead of one-shot capacities we could consider one-photon capacities, or more generally, ``unit-cost'' capacities. Instead of finite blocklength analyses of information processing tasks, we could perform finite cost analyses. The reverse substitution is also interesting. As we saw above, rather than cost-constrained capacities, we can consider blocklength-constrained capacities per unit cost. For probabilistic protocols or settings with feedback, we could imagine having an expected blocklength constraint instead of an expected cost constraint. Indeed, blocklength-cost duality allows us to take almost any question in information theory and ask its dual question by simply exchanging ``blocklength'' and ``cost.''

\bigskip

\noindent \textbf{Acknowledgements.} We are grateful to Emmanuel Abbe, Saikat Guha, Patrick Hayden, Alexander Holevo, Hideo Mabuchi, Sergio Verdu, and Tsachy Weissman for helpful conversations about this work. DD acknowledges support from the Stanford Graduate Fellowship program. DSP acknowledges support from the Army Research Office under grant number W911NF-16-1-0086. MMW acknowledges support from the US National Science Foundation and the Office of Naval Research. 

\bibliography{Ref}

\begin{thebibliography}{10}
\providecommand{\url}[1]{#1}
\csname url@samestyle\endcsname
\providecommand{\newblock}{\relax}
\providecommand{\bibinfo}[2]{#2}
\providecommand{\BIBentrySTDinterwordspacing}{\spaceskip=0pt\relax}
\providecommand{\BIBentryALTinterwordstretchfactor}{4}
\providecommand{\BIBentryALTinterwordspacing}{\spaceskip=\fontdimen2\font plus
\BIBentryALTinterwordstretchfactor\fontdimen3\font minus
  \fontdimen4\font\relax}
\providecommand{\BIBforeignlanguage}[2]{{%
\expandafter\ifx\csname l@#1\endcsname\relax
\typeout{** WARNING: IEEEtran.bst: No hyphenation pattern has been}%
\typeout{** loaded for the language `#1'. Using the pattern for}%
\typeout{** the default language instead.}%
\else
\language=\csname l@#1\endcsname
\fi
#2}}
\providecommand{\BIBdecl}{\relax}
\BIBdecl

\bibitem{book1991cover}
T.~M. Cover and J.~A. Thomas, \emph{Elements of Information Theory},
  2nd~ed.\hskip 1em plus 0.5em minus 0.4em\relax Wiley-Interscience, 2006.

\bibitem{pierce1981capacity}
J.~Pierce, E.~Posner, and E.~Rodemich, ``The capacity of the photon counting
  channel,'' \emph{IEEE Transactions on Information Theory}, vol.~27, no.~1,
  pp. 61--77, January 1981.

\bibitem{yin2017satellite}
J.~Yin, Y.~Cao, Y.-H. Li, S.-K. Liao, L.~Zhang, J.-G. Ren, W.-Q. Cai, W.-Y.
  Liu, B.~Li, H.~Dai \emph{et~al.}, ``Satellite-based entanglement distribution
  over 1200 kilometers,'' \emph{Science}, vol. 356, no. 6343, pp. 1140--1144,
  2017.

\bibitem{liao2017satellite}
S.-K. Liao, W.-Q. Cai, W.-Y. Liu, L.~Zhang, Y.~Li, J.-G. Ren, J.~Yin, Q.~Shen,
  Y.~Cao, Z.-P. Li \emph{et~al.}, ``Satellite-to-ground quantum key
  distribution,'' \emph{Nature}, vol. 549, no. 7670, p.~43, 2017.

\bibitem{reza1961introduction}
F.~M. Reza, \emph{An introduction to information theory}.\hskip 1em plus 0.5em
  minus 0.4em\relax Courier Corporation, 1961.

\bibitem{shannon1948mathematical}
C.~E. Shannon, ``A mathematical theory of communication,'' \emph{Bell System
  Technical Journal}, vol.~27, pp. 379--423, and 623--656, 1948.

\bibitem{pierce1978optical}
J.~Pierce, ``Optical channels: Practical limits with photon counting,''
  \emph{IEEE Transactions on Communications}, vol.~26, no.~12, pp. 1819--1821,
  December 1978.

\bibitem{golay1949note}
M.~J.~E. Golay, ``Note on the theoretical efficiency of information reception
  with {PPM},'' \emph{Proceedings of the Institute of Radio Engineers},
  vol.~37, no.~9, pp. 1031--1031, 1949.

\bibitem{gallager87energylimited}
R.~G. Gallager, ``Energy limited channels: Coding, multiaccess, and spread
  spectrum,'' Tech. Rep., 1987.

\bibitem{GHT11}
\BIBentryALTinterwordspacing
S.~Guha, J.~L. Habif, and M.~Takeoka, ``Approaching {Helstrom} limits to
  optical pulse-position demodulation using single photon detection and optical
  feedback,'' \emph{Journal of Modern Optics}, vol.~58, no. 3-4, pp. 257--265,
  December 2011, arXiv:1001.2447. [Online]. Available:
  \url{http://dx.doi.org/10.1080/09500340.2010.533204}
\BIBentrySTDinterwordspacing

\bibitem{guha2011structured}
S.~Guha, ``Structured optical receivers to attain superadditive capacity and
  the {Holevo} limit,'' \emph{Physical Review Letters}, vol. 106, no.~24, p.
  240502, June 2011, arXiv:1101.1550.

\bibitem{guha2011quantum}
S.~Guha, Z.~Dutton, and J.~H. Shapiro, ``On quantum limit of optical
  communications: concatenated codes and joint-detection receivers,'' in
  \emph{Proceedings of the 2011 IEEE International Symposium on Information
  Theory}.\hskip 1em plus 0.5em minus 0.4em\relax IEEE, 2011, pp. 274--278,
  arXiv:1102.1963.

\bibitem{DBEM11}
S.~Dolinar, K.~M. Birnbaum, B.~I. Erkmen, and B.~Moision, ``On approaching the
  ultimate limits of photon-efficient and bandwidth-efficient optical
  communication,'' in \emph{2011 International Conference on Space Optical
  Systems and Applications (ICSOS)}, May 2011, pp. 269--278, arXiv:1104.2643.

\bibitem{verdu1990channel}
S.~Verdu, ``On channel capacity per unit cost,'' \emph{IEEE Transactions on
  Information Theory}, vol.~36, no.~5, pp. 1019--1030, September 1990.

\bibitem{csiszar2007limit}
I.~Csisz{\'a}r, F.~Hiai, and D.~Petz, ``Limit relation for quantum entropy and
  channel capacity per unit cost,'' \emph{Journal of Mathematical Physics},
  vol.~48, no.~9, p. 092102, September 2007, arXiv:0704.0046.

\bibitem{holevo2003entanglement}
A.~S. Holevo, ``Entanglement-assisted capacity of constrained channels,'' in
  \emph{First International Symposium on Quantum Informatics}.\hskip 1em plus
  0.5em minus 0.4em\relax International Society for Optics and Photonics, 2003,
  pp. 62--69, arXiv:quant-ph/0211170.

\bibitem{H04}
\BIBentryALTinterwordspacing
------, ``Entanglement-assisted capacities of constrained quantum channels,''
  \emph{Theory of Probability \& Its Applications}, vol.~48, no.~2, pp.
  243--255, July 2004, arXiv:quant-ph/0211170. [Online]. Available:
  \url{http://dx.doi.org/10.1137/S0040585X97980415}
\BIBentrySTDinterwordspacing

\bibitem{S02}
P.~W. Shor, ``Additivity of the classical capacity of entanglement-breaking
  quantum channels,'' \emph{Journal of Mathematical Physics}, vol.~43, no.~9,
  pp. 4334--4340, 2002, arXiv:quant-ph/0201149.

\bibitem{Shirokov2006}
\BIBentryALTinterwordspacing
M.~E. Shirokov, ``The {Holevo} capacity of infinite dimensional channels and
  the additivity problem,'' \emph{Communications in Mathematical Physics}, vol.
  262, no.~1, pp. 137--159, February 2006, arXiv:quant-ph/0408009. [Online].
  Available: \url{http://dx.doi.org/10.1007/s00220-005-1457-8}
\BIBentrySTDinterwordspacing

\bibitem{holevo2008entanglement}
A.~S. Holevo, ``Entanglement-breaking channels in infinite dimensions,''
  \emph{Problems of Information Transmission}, vol.~44, no.~3, pp. 171--184,
  2008, arXiv:0802.0235.

\bibitem{U62}
H.~Umegaki, ``Conditional expectations in an operator algebra {IV} (entropy and
  information),'' \emph{Kodai Mathematical Seminar Reports}, vol.~14, no.~2,
  pp. 59--85, 1962.

\bibitem{J17}
M.~Jarzyna, ``Classical capacity per unit cost for quantum channels,'' April
  2017, arXiv:1704.06696.

\bibitem{stinespring1955positive}
W.~F. Stinespring, ``Positive functions on {C}*-algebras,'' \emph{Proceedings
  of the American Mathematical Society}, vol.~6, no.~2, pp. 211--216, 1955.

\bibitem{paulsen2002completely}
V.~Paulsen, \emph{Completely bounded maps and operator algebras}.\hskip 1em
  plus 0.5em minus 0.4em\relax Cambridge University Press, 2002, vol.~78.

\bibitem{holevo2013quantum}
A.~S. Holevo, \emph{Quantum systems, channels, information: a mathematical
  introduction}.\hskip 1em plus 0.5em minus 0.4em\relax Walter de Gruyter,
  2012, vol.~16.

\bibitem{giovannetti2004classical}
V.~Giovannetti, S.~Guha, S.~Lloyd, L.~Maccone, J.~H. Shapiro, and H.~P. Yuen,
  ``Classical capacity of the lossy bosonic channel: The exact solution,''
  \emph{Physical Review Letters}, vol.~92, no.~2, p. 027902, January 2004,
  arXiv:quant-ph/0308012.

\bibitem{GHG15}
V.~Giovannetti, A.~S. Holevo, and R.~Garcia-Patron,
  ``\BIBforeignlanguage{English}{A solution of {Gaussian} optimizer conjecture
  for quantum channels},'' \emph{\BIBforeignlanguage{English}{Communications in
  Mathematical Physics}}, vol. 334, no.~3, pp. 1553--1571, March 2015,
  arXiv:1312.2251.

\bibitem{giovannetti2014ultimate}
V.~Giovannetti, R.~Garcia-Patron, N.~J. Cerf, and A.~S. Holevo, ``Ultimate
  classical communication rates of quantum optical channels,'' \emph{Nature
  Photonics}, vol.~8, no.~10, pp. 796--800, September 2014, arXiv:1312.6225.

\bibitem{hiai1991proper}
F.~Hiai and D.~Petz, ``The proper formula for relative entropy and its
  asymptotics in quantum probability,'' \emph{Communications in Mathematical
  Physics}, vol. 143, no.~1, pp. 99--114, December 1991.

\bibitem{ogawa2000strong}
T.~Ogawa and H.~Nagaoka, ``Strong converse and {Stein's} lemma in quantum
  hypothesis testing,'' \emph{IEEE Transactions on Information Theory},
  vol.~46, no.~7, pp. 2428--2433, November 2000, arXiv:quant-ph/9906090.

\bibitem{schumacher1997sending}
B.~Schumacher and M.~D. Westmoreland, ``Sending classical information via noisy
  quantum channels,'' \emph{Physical Review A}, vol.~56, no.~1, p. 131, July
  1997.

\bibitem{holevo1998capacity}
A.~S. Holevo, ``The capacity of the quantum channel with general signal
  states,'' \emph{IEEE Transactions on Information Theory}, vol.~44, no.~1, pp.
  269--273, January 1998, arXiv:quant-ph/9611023.

\bibitem{wilde2013quantum}
M.~M. Wilde, \emph{Quantum information theory}, 2nd~ed.\hskip 1em plus 0.5em
  minus 0.4em\relax Cambridge University Press, 2017, arXiv:1106.1445.

\bibitem{winter2016tight}
A.~Winter, ``Tight uniform continuity bounds for quantum entropies: Conditional
  entropy, relative entropy distance and energy constraints,''
  \emph{Communications in Mathematical Physics}, vol.~1, no. 347, pp. 291--313,
  October 2016, arXiv:1507.07775.

\bibitem{yuen1993ultimate}
H.~P. Yuen and M.~Ozawa, ``Ultimate information carrying limit of quantum
  systems,'' \emph{Physical Review Letters}, vol.~70, no.~4, p. 363, January
  1993.

\bibitem{anshu2017one}
A.~Anshu, R.~Jain, and N.~A. Warsi, ``One shot entanglement assisted classical
  and quantum communication over noisy quantum channels: A hypothesis testing
  and convex split approach,'' February 2017, arXiv:1702.01940.

\bibitem{itit1999winter}
A.~Winter, ``Coding theorem and strong converse for quantum channels,''
  \emph{IEEE Transactions on Information Theory}, vol.~45, no.~7, pp.
  2481--2485, November 1999, arXiv:1409.2536.

\bibitem{thesis1999winter}
------, ``Coding theorems of quantum information theory,'' Ph.D. dissertation,
  Universit\"at Bielefeld, July 1999, arXiv:quant-ph/9907077.

\bibitem{devetak2005private}
I.~Devetak, ``The private classical capacity and quantum capacity of a quantum
  channel,'' \emph{IEEE Transactions on Information Theory}, vol.~51, no.~1,
  pp. 44--55, January 2005, arXiv:quant-ph/0304127.

\bibitem{1050633}
N.~Cai, A.~Winter, and R.~W. Yeung, ``Quantum privacy and quantum wiretap
  channels,'' \emph{Problems of Information Transmission}, vol.~40, no.~4, pp.
  318--336, October 2004.

\bibitem{wilde2016energy}
M.~M. Wilde and H.~Qi, ``Energy-constrained private and quantum capacities of
  quantum channels,'' September 2016, arXiv:1609.01997.

\bibitem{cmp2005dev}
I.~Devetak and P.~W. Shor, ``The capacity of a quantum channel for simultaneous
  transmission of classical and quantum information,'' \emph{Communications in
  Mathematical Physics}, vol. 256, no.~2, pp. 287--303, June 2005,
  arXiv:quant-ph/0311131.

\bibitem{el2013secrecy}
M.~El-Halabi, T.~Liu, and C.~N. Georghiades, ``Secrecy capacity per unit
  cost,'' \emph{IEEE Journal on Selected Areas in Communications}, vol.~31,
  no.~9, pp. 1909--1920, September 2013.

\bibitem{KRBM07}
\BIBentryALTinterwordspacing
R.~K\"onig, R.~Renner, A.~Bariska, and U.~Maurer, ``Small accessible quantum
  information does not imply security,'' \emph{Physical Review Letters},
  vol.~98, p. 140502, April 2007, arXiv:quant-ph/0512021. [Online]. Available:
  \url{http://link.aps.org/doi/10.1103/PhysRevLett.98.140502}
\BIBentrySTDinterwordspacing

\bibitem{shirokov2016lower}
M.~E. Shirokov and A.~S. Holevo, ``Lower semicontinuity of the entropic
  disturbance and its applications in quantum information theory,''
  \emph{Izvestiya: Mathematics}, vol.~81, no.~5, pp. 1044--1060, 2017,
  arXiv:1608.02203.

\bibitem{Lindblad1975}
G.~Lindblad, ``Completely positive maps and entropy inequalities,''
  \emph{Communications in Mathematical Physics}, vol.~40, no.~2, pp. 147--151,
  June 1975.

\bibitem{D09}
N.~Datta, ``Min- and max-relative entropies and a new entanglement monotone,''
  \emph{IEEE Transactions on Information Theory}, vol.~55, no.~6, pp.
  2816--2826, June 2009, arXiv:0803.2770.

\bibitem{U76}
A.~Uhlmann, ``The ``transition probability'' in the state space of a
  *-algebra,'' \emph{Reports on Mathematical Physics}, vol.~9, no.~2, pp.
  273--279, 1976.

\bibitem{tomamichel2015quantum}
M.~Tomamichel, \emph{Quantum Information Processing with Finite Resources:
  Mathematical Foundations}, ser. Springer Briefs in Mathematical
  Physics.\hskip 1em plus 0.5em minus 0.4em\relax Springer, 2016, vol.~5,
  arXiv:1504.00233.

\bibitem{PhysRevA.54.2629}
B.~Schumacher and M.~A. Nielsen, ``Quantum data processing and error
  correction,'' \emph{Physical Review A}, vol.~54, no.~4, pp. 2629--2635,
  October 1996, arXiv:quant-ph/9604022.

\bibitem{PhysRevA.55.1613}
S.~Lloyd, ``Capacity of the noisy quantum channel,'' \emph{Physical Review A},
  vol.~55, no.~3, pp. 1613--1622, March 1997, arXiv:quant-ph/9604015.

\bibitem{BKN98}
H.~Barnum, E.~Knill, and M.~Nielsen, ``On quantum fidelities and channel
  capacities,'' \emph{IEEE Transactions on Information Theory}, vol.~46, no.~4,
  pp. 1317--1329, July 2000, arXiv:quant-ph/9809010.

\bibitem{BNS98}
H.~Barnum, M.~A. Nielsen, and B.~Schumacher, ``Information transmission through
  a noisy quantum channel,'' \emph{Physical Review A}, vol.~57, no.~6, pp.
  4153--4175, June 1998, arXiv:quant-ph/9702049.

\bibitem{capacity2002shor}
P.~W. Shor, ``The quantum channel capacity and coherent information,'' in
  \emph{Lecture Notes, MSRI Workshop on Quantum Computation}, 2002.

\bibitem{S08}
G.~Smith, ``Private classical capacity with a symmetric side channel and its
  application to quantum cryptography,'' \emph{Physical Review A}, vol.~78,
  no.~2, p. 022306, August 2008, arXiv:0705.3838.

\bibitem{DW04}
I.~Devetak and A.~Winter, ``Relating quantum privacy and quantum coherence: An
  operational approach,'' \emph{Physical Review Letters}, vol.~93, no.~8, p.
  080501, August 2004, arXiv:quant-ph/0307053.

\bibitem{wang2012one}
L.~Wang and R.~Renner, ``One-shot classical-quantum capacity and hypothesis
  testing,'' \emph{Physical Review Letters}, vol. 108, no.~20, p. 200501, May
  2012, arXiv:1007.5456.

\bibitem{buscemi2010quantum}
F.~Buscemi and N.~Datta, ``The quantum capacity of channels with arbitrarily
  correlated noise,'' \emph{IEEE Transactions on Information theory}, vol.~56,
  no.~3, pp. 1447--1460, March 2010, arXiv:0902.0158.

\bibitem{datta2016second}
N.~Datta, M.~Tomamichel, and M.~M. Wilde, ``On the second-order asymptotics for
  entanglement-assisted communication,'' \emph{Quantum Information Processing},
  vol.~15, no.~6, pp. 2569--2591, June 2016, arXiv:1405.1797.

\bibitem{holevo2002sending}
A.~S. Holevo, ``Sending quantum information with {Gaussian} states,'' in
  \emph{Quantum Communication, Computing, and Measurement}.\hskip 1em plus
  0.5em minus 0.4em\relax Springer, 2002, pp. 75--82, arXiv preprint
  quant-ph/9809022.

\bibitem{holevo2001evaluating}
A.~S. Holevo and R.~F. Werner, ``Evaluating capacities of bosonic {Gaussian}
  channels,'' \emph{Physical Review A}, vol.~63, no.~3, p. 032312, February
  2001, arXiv:quant-ph/9912067.

\bibitem{takeoka2014capacity}
M.~Takeoka and S.~Guha, ``Capacity of optical communication in loss and noise
  with general quantum {Gaussian} receivers,'' \emph{Physical Review A},
  vol.~89, no.~4, p. 042309, April 2014, arXiv:1401.5132.

\bibitem{holevo2012quantum}
A.~S. Holevo and V.~Giovannetti, ``Quantum channels and their entropic
  characteristics,'' \emph{Reports on progress in physics}, vol.~75, no.~4, p.
  046001, April 2012, arXiv:1202.6480.

\bibitem{shirokov2012conditions}
M.~E. Shirokov, ``Conditions for coincidence of the classical capacity and
  entanglement-assisted capacity of a quantum channel,'' \emph{Problems of
  Information Transmission}, vol.~48, no.~2, pp. 85--101, 2012,
  arXiv:1202.3449.

\bibitem{WPG07}
M.~M. Wolf, D.~P\'erez-Garc\'\i{}a, and G.~Giedke, ``Quantum capacities of
  bosonic channels,'' \emph{Physical Review Letters}, vol.~98, no.~13, p.
  130501, March 2007, arXiv:quant-ph/0606132.

\bibitem{DSW18}
N.~Davis, M.~E. Shirokov, and M.~M. Wilde, ``Energy-constrained two-way
  assisted private and quantum capacities of quantum channels,'' \emph{Physical
  Review A}, vol.~97, no.~6, p. 062310, Jun. 2018, arXiv:1801.08102.

\bibitem{takeoka2015fundamental}
M.~Takeoka, S.~Guha, and M.~M. Wilde, ``Fundamental rate-loss tradeoff for
  optical quantum key distribution,'' \emph{Nature Communications}, vol.~5, p.
  5235, 2014, arXiv:1504.06390.

\bibitem{goodEW16}
K.~Goodenough, D.~Elkouss, and S.~Wehner, ``Assessing the performance of
  quantum repeaters for all phase-insensitive {Gaussian} bosonic channels,''
  \emph{New Journal of Physics}, vol.~18, p. 063005, 2016, arXiv:1511.08710v2.

\bibitem{lebedev1966information}
D.~S. Lebedev and L.~B. Levitin, ``Information transmission by electromagnetic
  field,'' \emph{Information and Control}, vol.~9, no.~1, pp. 1--22, 1966.

\bibitem{caves1994quantum}
C.~M. Caves and P.~D. Drummond, ``Quantum limits on bosonic communication
  rates,'' \emph{Reviews of Modern Physics}, vol.~66, no.~2, p. 481, April
  1994.

\bibitem{yuen1975optimum}
H.~Yuen, R.~Kennedy, and M.~Lax, ``Optimum testing of multiple hypotheses in
  quantum detection theory,'' \emph{IEEE Transactions on Information Theory},
  vol.~21, no.~2, pp. 125--134, March 1975.

\bibitem{bardhan2014strong}
B.~R. Bardhan, R.~Garcia-Patron, M.~M. Wilde, and A.~Winter, ``Strong converse
  for the classical capacity of all phase-insensitive bosonic {Gaussian}
  channels,'' \emph{IEEE Transactions on Information Theory}, vol.~61, no.~4,
  pp. 1842--1850, April 2015, arXiv:1401.4161.

\end{thebibliography}
\end{document}